\documentclass[11pt]{report} 

\pdfoutput=1

\usepackage{epigraph}
\usepackage{citehack} 

\usepackage{fancyhdr}
\usepackage{pdflscape}
\usepackage{multirow}
\usepackage[utf8]{inputenc} 
\usepackage{geometry}
\usepackage{amsmath, amsthm,amssymb,amsfonts,latexsym}

\newtheorem{theorem}{Theorem}[chapter]
\newtheorem{proposition}[theorem] {Proposition}

\newtheorem{lemma}[theorem]{Lemma}
\newtheorem{corollary}[theorem]{Corollary}

\theoremstyle{definition}
\newtheorem{definition}[theorem]{Definition}

\newtheorem{question}[theorem]{Open Question}

\newcommand{\NN}{{\mathbb{N}}}

\newcommand{\QQ}{{\mathbb{Q}}}
\newcommand{\ZZ}{{\mathbb{Z}}}

\newcommand{\SI}[1]{\Sigma^0_{#1}}
\newcommand{\PI}[1]{\Pi^0_{#1}}

\newcommand{\bi}{\begin{itemize}}
\newcommand{\ei}{\end{itemize}}
\newcommand{\bc}{\begin{center}}
\newcommand{\ec}{\end{center}}

\newcommand{\ria}{\rightarrow}

\newcommand{\la}{\langle}
\newcommand{\ra}{\rangle}

\newcommand{\DE}[1]{\Delta^0_#1}

\newcommand{\lra}{\leftrightarrow}
\newcommand{\LR}{\Leftrightarrow}
\newcommand{\RA}{\Rightarrow}
\newcommand{\LA}{\Leftarrow}

\usepackage[square, sort, numbers]{natbib}
\usepackage{hyperref}
\usepackage{hypernat}
\usepackage[pdftex]{graphicx}

\newcommand{\doctitle}{Computable Component-wise Reducibility}
\newcommand{\docauthor}{Egor Ianovski}
\newcommand{\docdate}{\today}
\title{\doctitle{}}
\author{\docauthor{}}
\date{\docdate{}}
\hypersetup{pdftitle = {\doctitle{} | \docauthor{}},
            pdfauthor = {\docauthor{}},
            pdfborder = {0 0 0.5}}

\usepackage{titlesec}
\titleformat{\chapter}[hang]{\bfseries \huge}{\thechapter}{2pc}{}

\usepackage{tikz}

\usepackage[OT2,OT1]{fontenc}
\newcommand\cyr{%
\renewcommand\rmdefault{wncyr}%
\renewcommand\sfdefault{wncyss}%
\renewcommand\encodingdefault{OT2}%
\normalfont
\selectfont}
\DeclareTextFontCommand{\textcyr}{\cyr}

\usepackage{citehack}

\usepackage[newzealand]{babel}
\addto{\captionsnewzealand}{}

\begin{document}

\begin{titlepage}

\begin{center}

\includegraphics[width=0.15\textwidth]{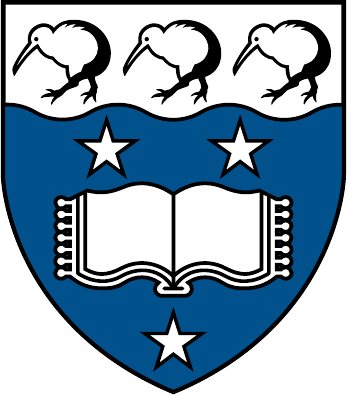}\\[1.5cm]    

\textsc{\LARGE University of Auckland}\\[1cm]

\textsc{\Large Department of Computer Science}\\[0.2cm]

\textsc{\Large Department of Philosophy}\\[0.8cm]

\hrule
\vspace{0.3cm}

{ \Large Computable Component-wise Reducibility}

\vspace{0.3cm}
\hrule
\vspace{0.9cm}
\begin{minipage}{0.4\textwidth}
\begin{flushleft} \large
\emph{Author:}\\
Egor Ianovski
\end{flushleft}
\end{minipage}
\begin{minipage}{0.4\textwidth}
\begin{flushright} \large
\emph{Supervisor:} \\
Dr.~Andr\'{e} Nies
\end{flushright}
\end{minipage}

\vfill

{\large Submitted in fulfilment of the requirements of the degree of}\\
{\large MSc in Logic and Computation. Last updated \today.}

\end{center}

\end{titlepage}

\setcounter{page}{1}
\renewcommand{\thepage}{\roman{page}}

\begin{abstract}
\thispagestyle{plain} 
We consider equivalence relations and preorders complete for various levels of
the arithmetical hierarchy under computable, component-wise reducibility. We
show that implication in first order logic is a complete preorder for $\SI 1$,
the $\le^P_m$ relation on EXPTIME sets for $\SI 2$ and the embeddability of
computable subgroups of $(\QQ,+)$ for $\SI 3$. In all cases, the symmetric
fragment of the preorder is complete for equivalence relations on the same
level. We present a characterisation of $\PI 1$ equivalence relations which
allows us to establish that equality of polynomial time functions and inclusion
of polynomial time sets are complete for $\PI 1$ equivalence relations and
preorders respectively. We also show that this is the limit of the enquiry: for
$n\geq 2$ there are no $\PI n$ nor $\DE n$-complete equivalence relations.
\end{abstract}

\setcounter{page}{2}
\section*{Acknowledgements}

I would like to extend my gratitude to the following people.

First and foremost to my supervisor Andr\'e Nies for
introducing me to the field of Recursion Theory. Though my steps did falter, and
the vistas were bewildering at times, I feel there could have been no better
reward for my labour than the enriched understanding it brought of the world of
computability.

To Alexander Gavruskin and Jiamou Liu for their
suggestions and commentary on my work in the periods whilst Andr\'e was away.
Being able to discuss research with others is an invaluable anchor on sanity in
periods when one's studies take one deeper, and deeper into theory and
abstraction.

To my co-authors Russell Miller and Keng Meng Ng. Much of the results in this
work will later appear in \cite{Ianovski2012}, and a fascinating result of
theirs I will present in the penultimate section here.

To my family for the moral and material support they provided through the
breadth of my studies. Ever they had more faith in me than I did myself, and I
can only hope that one day I may justify it.

Finally, to the creators and maintainers of the various \LaTeX\ resources and
guides throughout the internet. Without them mathematics would be a lot less
readable.

\tableofcontents

\newpage

\setcounter{page}{1}
\renewcommand{\thepage}{\arabic{page}}

\chapter{Introduction}

\begin{center}
\begin{tabular}{rl}
\textsc{Lepidus}&
What manner o' thing is your crocodile?\\
\textsc{Antony}&
It is shaped, sir, like itself; and it is as broad as it hath
breadth:\\
&it is just so high as it is, and moves with its own
organs:\\
&it lives by that which nourisheth it, and the elements
once out of it,\\
&it transmigrates.\\
\end{tabular}

\vspace{0.3cm}
\hfill-William Shakespeare, \emph{Antony and Cleopatra}
\end{center}
\vspace{1cm}

\noindent The concept of reducibility plays a prominent role throughout the
field of computer science, and for good reason; questions like ``What is
computable", ``tractable", ``complex" are notoriously difficult. Worse, they are
not even mathematical: at best we find ourselves in empirical science, more
probably in philosophy. And philosophical disputes have the unfortunate tendency
to sustain themselves for centuries with no resolution.

Reducibility allows us
to handle this issue with grace. While we cannot, with any certainty, say
anything objective about the difficulty of boolean satisfiability in light of
the yet unresolved relation regarding P and NP, and the more fundamental
question as to whether P indeed characterises the tractable problems, we can be
confident that this difficulty lies within a polynomial factor of graph
colouring. Some may deny the always-halting Turing machine as the limit of
computability, but the fact that if such a machine could solve first order
provability it could also solve the halting problem is uncontroversial. The more
determined constructivists among us may reject the axiom of choice as having a
place in mathematics, but in doing so we can all agree that they reject
Zorn's lemma as well.

Given the flexibility of the concept it is not surprising that it found its way
into many mathematical disciplines. Any field dealing with structures and
equipped with a notion of complexity can quite naturally seek to use the one to
impose some hierarchy on the other. Over and above its utility in mathematics,
the concept of reducibility has been described as a characteristic feature of
Computer Science. In particular the notion of NP-completeness is seen as an
important intellectual export of the field to other disciplines
(\cite{Papadimitriou1997}).

The specific form of reducibility studied here is what we term computable
component-wise reducibility: a computable function acts on each component
of a tuple separately; the objects under consideration are predominantly
equivalence relations, although often we will obtain the results via the more
general case of preorders.

As we will see in the literature review, the degree structure of equivalence
relations under this form of reducibility is rich and sometimes baffling. An
assault on the mysteries concealed within lies outside the scope of our enquiry.
Instead, we content ourselves with the much more limited goal of studying the
maximum elements of this structure: the complete relations for various levels of
the arithmetical hierarchy.

Our goal is to find equivalence relations and preorders that are not merely
complete, but in some sense natural, at least to a mathematician. Of course,
mulling too much over what it means to be ``natural" will pull us back into the
deep, dark woods of philosophical debate, so instead we present relations from a
variety of disciplines in the hope that the reader may find at least one
palatable. For a logician we have implication in first order logic, a computer
scientist we will serve polynomial time reducibility and to a
mathematician we hope to sell the embeddability of computable groups.
Ultimately, however, the reader should bear in mind that these are all just
subsets of $\NN\times\NN$ in various fineries. That the natural numbers can
display such richness and variety is perhaps the most fascinating and rewarding
aspect of our discipline. If the rest of this work is lost in formalism and
technicalities, we hope that at the very least the reader will find time to
reflect on the most beautiful mathematical structure of them all.

\section{Outline}

There are five parts to the sequel. In Section~\ref{sec:Preliminaries} we define
precisely what we mean by component-wise
reducibility, and offer some motivation for the choice. We introduce the
terminology and notation we will use and prove some
basic properties about the subject matter. In Section~\ref{sec:literature} we
present an
overview of some previous results in the field, and compare the computable case
to component-wise reducibilities in the related disciplines of complexity and
descriptive set theory. In
Section~\ref{sec:Si0nrelations} we begin our search for $\SI n$-complete
equivalence
relations. We present examples of complete relations and preorders for $\SI 1$,
$\SI 2$ and $\SI 3$, based in logic, complexity theory and group
theory respectively. In
Section~\ref{sec:Pi0nrelations} we construct a $\PI 1$-complete
equivalence relation and preorder and prove that this is as high as we can go:
we show that for $n\geq 2$,
no $\PI n$-complete equivalence relations exist. We conclude in
Section~\ref{sec:Conclusion}.

\chapter{Preliminaries}\label{sec:Preliminaries}

The subject matter of this thesis is what we will term computable
component-wise reducibility, and denote by $\le$. To state the definition,
given two binary relations $A$ and $B$, we say that $A\le B$ if there exists a
computable function $f$ such that:
\begin{equation}\label{eq:component_wise}
(x,y)\in A\iff(f(x),f(y))\in B.
\end{equation}
For contrast, recall that the standard notion of $m$-reducibility extended to
pairs would be:
\begin{equation}\label{eq:m_reducible}
(x,y)\in A\iff f(x,y)\in B.
\end{equation}

Neither the name nor the notation for this form of reducibility is standard;
we will mention some of the other names used in the literature in
Section~\ref{sec:literature}. While the definition given can be applied to any
binary relation, our main focus will be on equivalence relations and preorders.
Note also that we will also use $\le$ to denote the standard less-than-or-equal
relation on the integers. This should cause no ambiguity given the context.

This choice of reducibility and the focus on equivalence relations is not
arbitrary. In Section
\ref{sec:Intuition} we will show how this definition is
equivalent to concepts in other areas of mathematics. In Section
\ref{sec:Basic_results} we will prove some basic results about this manner of
reducibility, both to offer a glimpse as to how it differs from the standard
$m$-reduction of computability theory and to establish some facts that will be
useful in the later sections. First, however, we must establish notation
and terminology.

\section{Notation and terminology}\label{sec:Notation}

Our model of computation, where needed, will be a one-way infinite
Turing machine. Most results will not, however, explicitly invoke this and will
be based on a tacit acceptance of the Church-Turing Thesis.

As usual we assume some enumeration of all such machines and use $\varphi_e$ to
denote the $e$th partial recursive (p.r.)\ function arising from this
enumeration and
$W_e$ the $e$th recursively enumerable (r.e.)\ set. Likewise, $\varphi_{e,t}$
and
$W_{e,t}$ represent the $e$th function and set when the computation time is
restricted to $t$ steps. As every machine induces a function we may sometimes
use $M(w)$ to denote the output of $M$ on input $w$. If $M$ is a decider, then
we will say $M(w)=1$ to mean $M$ accepts $w$ and $0$ if it rejects. For an
oracle machine, $M(A;w)$ is used to mean Machine $M$ on input $w$ with access to
$A$ as an oracle.

As our subject matter is discrete, we will use $[a,b]$ to denote all integers
(rather than real
numbers) between $a$ and $b$ inclusive.

For an equivalence relation $E$, $[x]_E$
will be used to denote the equivalence
class of $x$ under $E$ and $\min[x]_E$ the least element (with respect to the
standard order on the natural numbers) of the equivalence class of $x$. If $E$
is clear from context we may omit the subscript. Where we are dealing with a
preorder $P$, we will reserve the words ``larger" and ``smaller" for the
standard order on the natural numbers, instead using ``above" and ``below" for
the order induced by $P$.

Angular brackets $\la\cdot\ra$ are used to denote an encoding function, mapping
$\cdot$ injectively to some natural number. If we do not mention the
requirements on the encoding explicitly, it is to be taken that any computable
encoding suffices. Strictly speaking, as the sets $W_e$ are subsets of $\NN$ we
should really use $\la x,y\ra$ for any binary relation considered here. However
we will often tacitly extend the notion of a r.e.\ set to be a subset of
$\NN\times\NN$, as such an extension does not affect any of our results. This
allows us to reserve the angular brackets for instances where the encoding
function plays some necessary role, or simply to reduce stacked parentheses.

Standard notation is used for the arithmetical hierarchy. We obtain a $\SI n$
relation by making an existential query over a $\PI {n-1}$ relation, a $\PI n$
from
the complement of a $\SI n$ and $\DE n=\SI n\cap\PI n$. The base case are
the computable relations $\SI 0=\DE 1=\PI 0$. We do not introduce special
notation to distinguish between the $\SI n$ $(\PI n,\DE n)$ sets, equivalence
relations or preorders. The specific meaning should be clear from context.

\section{Intuition}\label{sec:Intuition}

Component-wise reducibility is a clear analogue of embeddability in model
theory. Recall that given $A\subseteq U^2$ and $B\subseteq V^2$, we say that the
structure $(U,A)$ is embeddable in $(V,B)$ if there exists a function $f:U\ria
V$ satisfying:
\begin{equation}
(x,y)\in A \iff (f(x),f(y))\in B.
\end{equation}
If we restrict the setting so that $f$ is required to be computable, and the
domains of the structures are $\NN$ (or, given a suitable encoding, some
countable set), we obtain \eqref{eq:component_wise}: $(\NN,A)$ is computably
embeddable in $(\NN,B)$ if and only if $A\le B$.

Herein lies a crucial aspect of component-wise reducibility. Whereas
$m$-reducibility is only concerned with computability, a component-wise
reduction also implies a structural similarity between the relations. For
example, if $S$ and $R$ are two equivalence relations, and $S\le R$ via $f$,
then $g:\NN/S\ria\NN/R$ given by $[x]\mapsto[f(x)]$ is one to one.

The concept also arises from category theory (\cite{Andrews2010}), via the
theory
of enumerations of Ershov
(\cite{Ershov1969}). An enumeration is a function mapping the natural numbers
onto
some set. Any such $a:\NN\ria A$ defines an equivalence relation, $\{(x,y)\ |\
a(x)=a(y)\}$, and any equivalence relation over $\NN^2$ can be defined by some
enumeration: namely, $x\mapsto\min[x]$. It thus makes sense to class
enumerations by the complexity of the equivalence relation they define. In the
terminology of \cite{Ershov1969}, the \emph{positive} enumerations are those
which define $\SI 1$ equivalence relations and the \emph{negative} define $\PI
1$, and there is no difficulty in likewise associating a class of enumerations
with every level of the arithmetical hierarchy.

A \emph{morphism of enumerations} is a computable function $f:\NN\ria\NN$ such
that for
two enumerations $a$ and $b$, $a\circ f=b$. This gives us another means to view
component-wise reducibility: $f$ is a morphism from $a$ to $b$ if and only if
$f$ is a reduction between the induced equivalence relations, $E_a$ and $E_b$,
as
shown in Figure~\ref{fig:enum_morph}. Note that this induces the existence of a
unique $g:A\ria B$ such that the resulting diagram commutes. In fact $g$ is
precisely the injection from $\NN/E_a$ to $\NN/E_b$ we have seen already.

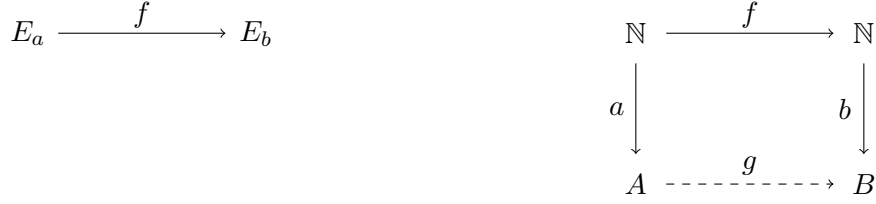
\begin{figure}
\begin{center}
\begin{tikzpicture}
\node at (0,0) {$E_a$};
\node at (3,0) {$E_b$};

\draw[->] (0.4,0)--(2.6,0);

\node at (1.5,0.25) {$f$};

\node at (8,0) {$\NN$};
\node at (11,0) {$\NN$};

\draw[->] (8.4,0)--(10.6,0);

\node at (9.5,0.25) {$f$};

\node at (8,-2) {$A$};
\node at (11,-2) {$B$};

\draw[->] (8,-0.4)--(8,-1.6);
\draw[->] (11,-0.4)--(11,-1.6);

\node at (7.75,-1) {$a$};
\node at (10.75,-1) {$b$};

\draw[->,dashed] (8.4,-2)--(10.6,-2);

\node at (9.5,-1.75) {$g$};
\end{tikzpicture}
\end{center}
\caption{A component-wise reduction is also a morphism of enumerations}
\label{fig:enum_morph}
\end{figure}

If we thus fix a category of enumerations of a given level in the arithmetical
hierarchy, the terminal object in this category is an enumeration with a
morphism from every other enumeration; an equivalence relation to which there
exists a reduction from any other equivalence relation: precisely the complete
equivalence relation that is the focus of our study.

\section{Basic results}\label{sec:Basic_results}

In this section we will prove a few simple results about the behaviour of
computable component-wise reducibility for two reasons. The first is to convince
ourselves
that this notion of reducibility is, indeed, substantially different from the
standard $\le_m$ or $\le_1$ reducibilities used in computability theory, and the
second
is to see that there nevertheless exist some interesting connections between
them.

\begin{proposition}\label{prop:comp_implies_set}
For binary relations $A$ and $B$, $A\le B\RA A\le_m B$, but not vice versa.
\end{proposition}
\begin{proof}
By definition, $A\le B$ implies the existence of a computable $f$ satisfying:
$$(x,y)\in A\qquad\iff\qquad(f(x),f(y))\in B.$$
Define $f'$ by $f'(x,y)=(f(x),f(y))$. As $f$ is computable, clearly so is $f'$.
This gives us:
$$(x,y)\in A\qquad\iff\qquad f'(x,y)\in B,$$
which is to say, $A\le_m B$.

The simplest way to see that $A\le_m B$ does not imply $A\le B$ is to recall the
observations of Section~\ref{sec:Intuition}: $A\le B$ implies that in the
corresponding computable structures $A$ is embeddable in $B$, and as such $B$
restricted to the image of $f$ satisfies the same
relational properties as $A$. It follows that, for example, an $A$ that contains
a reflexive pair
cannot be component-wise reduced to a $B$ that is everywhere irreflexive.
Thereby if we use $K$ to denote the halting problem, namely:
$$K=\{(e,w)\ |\ \exists y,t\ \varphi_{e,t}(w)=y\},$$
and if $e$ is the index of some never-halting machine, $K\times K\nleq
K\times\{(e,w)\}$. However, clearly $K\times K\le_m
K\times\{(e,w)\}$ via $((i,v),(i',v'))\mapsto((j,\la v,v'\ra),(e,w))$, where $j$
is the machine that will first run machine $i$ on $v$ then machine $i'$ on $v'$.
\end{proof}

A corresponding statement fails for $\le_1$. We can construct binary relations
where $\le$ and $\le_m$ coincide, and as $\le_m$ does not imply $\le_1$, it
follows that component-wise reducibility as defined here is incomparable with
1-reducibility.

\begin{proposition}
There exist binary relations such that $A\le B$ but $A\nleq_1 B$.
\end{proposition}
\begin{proof}
Let $C$ and $D$ be any two sets such that $C\le_m D$ but $C\nleq_1 D$. Shift the
elements of these sets up by one, creating $C^+=\{x+1\ |\ x\in C\}\cup\{0\}$ and
likewise with $D$.
Let $A=C^+\times \{0\}$, $B=D^+\times \{0\}$.

Clearly, $A\le B$: if $f$ is the $m$-reduction from $C$ to $D$, let $f^+$ be
given by the maps $x+1\mapsto f(x)+1$ and $0\mapsto 0$. Observe that $(x,y)\in
A$
if and only if $(f^+(x),f^+(y))\in B$.
However, $A\nleq_1B$. Suppose to the contrary that there exists a one to one $g$
satisfying:
$$
(x,y)\in A\qquad\iff\qquad g(x,y)\in B.$$

Let $g_1$ be the first component of $g$. As the second component is always 0, it
must be the case that $g_1$ is one to one, and it is clearly computable. This
gives us:
$$
x\in C\qquad\iff\qquad g_1(x)-1\in D,$$
which contradicts the assumption on $C$ and $D$.
\end{proof}

As a useful starting point in the investigation of the complexity of sets under
component-wise reducibility, we will demonstrate that there exist $\SI n$ and
$\PI
n$-complete relations for every $n$. On the contrary, for $n\geq 2$ there are no
$\DE n$-complete relations. Note that here we are
speaking of general binary relations, so this in no way contradicts the result
we will later present in Theorem~\ref{thm:noPi0nrelation}.

\begin{proposition}\label{prop:Si0n_complete_relations}
For every $n$, there is a $\SI n$ and a $\PI n$ complete relation. However,
there are no $\DE n$-complete relations for $n\geq 2$. 
\end{proposition}
\begin{proof}
Let us first consider the case with $n=1$. As $\SI 1$ relations are precisely
those enumerated by some machine, we can fix some encoding of pairs and treat
$W_e$ as an r.e subset of $\NN^2$. It then makes sense to speak of the $e$th
$\SI 1$ relation:
precisely that which is enumerated by $W_e$.

The $e$th $\SI 1$ relation for any $e$ can then be reduced to
$S=\{(\la x,e\ra,\la y,i\ra)\ |\ e=i\wedge (x,y)\in W_e)\}$ via the function
$x\mapsto\la x,e\ra$. As such this relation is $\SI 1$ complete.

As with sets, the complement of a $\SI 1$-complete relation is
complete for $\PI 1$. This can be seen by considering that for any $R$ in $\SI
1$:
$$(x,y)\in \overline{R}\iff (x,y)\notin R\iff (f(x),f(y))\notin S\iff
(f(x),f(y))\in\overline{S}.$$

For $n>1$ the argument is identical, except in this case a $\SI n$ relation is
one enumerated by a machine with access to a $\PI {n-1}$ oracle. Once we fix an
enumeration of the sets generated by those, the rest follows.

For the $\DE n$ case, recall that for $n\geq 2$ there are no $\DE n$-complete
sets (\cite{Ershov1970}). We
can see this by letting $A$ be any $\DE n$ set, $n\geq 2$, and letting $B$ be as
follows:
$$B=\{e\ |\ \exists y,t\ \varphi_{e,t}(e)=y\ria\varphi_e(e)\notin A\}.$$
That is, $B$ consists of the indices of partial recursive functions that either
diverge on their own index or whose value on their own index is not in $A$. Note
that $B$ is $\DE n$: determining whether a partial recursive function diverges
on a given input is $\PI 1\subseteq\DE n$, and since $A$ is a $\DE n$ set
determining whether an element is not in $A$ is $\DE n$.

We claim that
$B\nleq_m A$: suppose to the contrary that there exists a computable $f$ such
that $x\in B$ if and only
if
$f(x)\in A$. As $f=\varphi_e$ for some $e$, we can consider the behaviour of $f$
on its own index. Suppose $e\in B$. By the definition of $B$, either $f(e)$
diverges, which is impossible as $f$ is computable, or $f(e)\notin A$, which
contradicts our assumption on $f$. It follows that $e\notin B$. However, that is
impossible as it would necessitate that $f(e)\notin A$ and so, by the definition
of $B$, $e\in B$.

This shows that for $n\geq 2$ there can therefore be no $\DE n$-complete set,
so by
Proposition~\ref{prop:comp_implies_set} there can be no $\DE n$-complete
relation.
\end{proof}

In the sequel, we will not concern ourselves with arbitrary binary relations.
Our objects of study are equivalence relations and preorders, and it is of
interest that the argument in Proposition~\ref{prop:Si0n_complete_relations} can
be lifted to establish the existence of equivalence relations and preorders
complete for $\SI n$.

Note that for $W_e\in\SI 1$, if we abuse notation then the
transitive closure of
$W_e$ can be expressed as:
\begin{equation}
W_e\cup\{(x,y)\ |\ \exists k\ ((x,z_1)\in W_e\wedge(z_1,z_2)\in
W_e\wedge\dots\wedge (z_k,y)\in W_e)\}.
\end{equation}
To make the statement more precise, $k$ would need to be a computable encoding
of a finite tuple of indices. Regardless, as we are merely adding an additional
existential quantifier over a finite conjunction, the resulting
relation is still $\SI 1$.

Since every $\SI 1$ equivalence relation
arises from the
transitive closure of some symmetric and reflexive $\SI 1$ relation, it makes
sense to speak of the $e$th $\SI 1$ equivalence relation: for a given $W_e$ add
the reflexive and symmetric pairs and take the transitive closure. We can then
modify the proof of
Proposition~\ref{prop:Si0n_complete_relations} to obtain a $\SI n$-complete
equivalence relation. This gives us the same construction of a universal
equivalence
relation as in
Proposition 2.10 of \cite{Coskey2012}.

We can likewise speak of an $e$th $\SI 1$ preorder:
add the reflexive pairs to $W_e$ and take the transitive closure. A complete
preorder exists by the same construction.

There are other ways we can derive complete equivalence relations from more
general classes. Of use to us will be the simple observation that every
equivalence relation arises as the symmetric fragment of some preorder, and the
symmetric fragment of a $\SI n$, $\PI n$ or $\DE n$ preorder is still $\SI n$,
$\PI n$ or $\DE n$
respectively.

\begin{proposition}\label{prop:preorders}
The existence of a $\SI n$, $\PI n$ or $\DE n$-complete preorder implies the
existence of a $\SI n$, $\PI n$ or $\DE n$-complete equivalence relation
respectively.
\end{proposition}
\begin{proof}
Let $S$ be a preorder complete for some level of the arithmetical hierarchy. Let
$S_{sym}$ be the symmetric fragment of $S$; that is, $\{(x,y)\ |\ Sxy\wedge
Syx\}$.
We claim that $S_{sym}$ is the required equivalence relation.

It is straightforward to verify that $S_{sym}$ is indeed a $\SI n$ ($\PI n$,$\DE
n$)
equivalence relation: it is symmetric, transitive and reflexive, and $\SI n$
($\PI n$,$\DE n$) languages are closed under intersection.

To see that it is complete, note that as any $\SI n$ ($\PI n$,$\DE n$)
equivalence
relation $E$ is
also a preorder, there exists a computable $f$ satisfying:
$$Exy\qquad\iff\qquad Sf(x)f(y).$$
However as $E$ is symmetric, $S$ has to preserve that property under the image
of $f$ and we get:
$$
Exy\iff Exy\wedge Eyx\iff Sf(x)f(y)\wedge Sf(y)f(x)\iff S_{sym}f(x)f(y).$$
The proposition follows.
\end{proof}

As our enquiry was motivated by equivalence relations, we will end this
section with an
interesting observation between the connection of component-wise and
$m$-reducibility with respect to equivalence classes. This will later be seen in
the construction of a $\PI 1$-complete equivalence relation in
Section~\ref{sec:Pi0nrelations}.

\begin{proposition}\label{prop:relations_minelements}
If $E$ is a $\PI 1$ $(\SI 1)$ equivalence relation, then $A=\{x\ |\ x=\min[x]\}$
is $\SI 1$ $(\PI 1)$.
\end{proposition}
\begin{proof}
Let $E\in\PI 1$. Observe that $x$ is the minimum element of its equivalence
class if and only if for all $y<x$, $(x,y)\notin E$. $E$ is the complement of
an r.e set, so $E=\NN^2\setminus W_e$ for some $e$. As such:
$$
A=\{x\ |\ \exists t\ [0,x]^2\setminus W_{e,t}=\{(x,x)\}\}.$$
For $E\in\SI 1$, $E=W_e$ so we can simply pick out the elements $y<x$ such that
$(x,y)$ never enters $W_e$:
$$A=\{x\ |\ \forall t\forall y\!<\!x\ (x,y)\notin W_{e,t}\}.$$\end{proof}

\subsection{A \texorpdfstring{$\mathbf{\DE 1}$}{Delta 01}-complete preorder}

While the existence of a $\DE 1$-complete equivalence relation has been known
for some time, the case of preorders has as yet escaped attention. We will hence
here show how a complete preorder can be obtained in the computable case.

As opposed to subsequent sections, we are not here concerned with motivating the
result as mathematically natural. Instead we focus purely on the technical
aspect of the problem. The main difficulty lies in defining a preorder general
enough to allow any
computable preorder to be embedded in it. Intuitively, we address this by
piecing every computable preorder together. However as there is no effective
enumeration of $\DE 1$ sets, let alone the transitive ones, we in fact have to
consider approximations of computable preorders in a given number of steps. As
such we will define our preorder by letting $(x,e,t)\preceq (y,i,r)$ whenever
$x=y,e=i,t=r$, or the following hold:
\begin{itemize}
\item
$e=i$.
\item
The $e$th machine halts on all pairs in $[0,x]$ within $t$ steps, and all pairs
in $[0,y]$ within $r$ steps, and the pairs accepted constitute a preorder.
\item
The pair $(x,y)$ is accepted by the $e$th machine.
\end{itemize}
A relation so defined is clearly reflexive and computable. Once we establish
transitivity, showing it is complete will be straightforward.

\begin{theorem}
The preorder $\preceq$ described above is complete for $\DE 1$ preorders.
\end{theorem}
\begin{proof}
First let us verify that $\preceq$ is transitive. Suppose $(x,e,t)\preceq
(y,i,r)\preceq (z,k,s)$ for distinct triples. For contradiction, let
$(x,e,t)\npreceq (z,k,s).$ The possibility that $e\neq k$ is ruled out because
$e=i$ and $i=k$. The machine in question also needs to halt on all pairs in
$[0,x]$
and $[0,z]$ in order to have $(x,e,t)\preceq
(y,i,r)$ and $(y,i,r)\preceq
(z,k,s)$, so $(x,e,t)\npreceq (z,k,s)$ cannot stem from the failure of the
machine to define a preorder in the required number of steps. As such the only remaining
possibility is that the machine rejects the pair
$(x,z)$. But that is clearly impossible as the machine accepts $(x,y)$ and
$(y,z)$, and the pairs accepted in $[0,\max(x,y,z)]$ constitute a preorder.

It remains to define the reduction $f$, and the choice is clear: given a
computable preorder $P$ and an element $x$, $f$ evaluates $P$ on all pairs in
$[0,x]$ and maps $x$ to $(x,e,t)$, where $e$ is the index of the machine
deciding $P$ and $t$ is the maximum of the time required to decide any of the
pairs.
\end{proof}
\begin{corollary}
There exists a $\DE 1$-complete equivalence relation.
\end{corollary}
\begin{proof}
By application of Proposition~\ref{prop:preorders}. We will see a much simpler
example in the next section.
\end{proof}

\chapter{Literature Review}\label{sec:literature}

We can trace the study of component-wise reductions through three fields,
distinguished by the requirements placed on the reduction in each case: by Borel
functions in descriptive set theory, polynomial time functions in complexity
theory and computable functions in computability theory. This is not to say that
these are the only component-wise reducibilities present in the literature; for
instance hyperarithmetical reductions are considered by Fokina and Friedman
(\cite{Fokina2012:2}),
effectively Borel ($\Delta^1_1$) reductions by Fokina, Friedman and T\"ornquist
(\cite{Fokina2010}) and infinite-time computable reductions by Coskey and
Hamkins \cite{Coskey2011}.

We choose to focus on the three selected here due to their affinity in spirit to
the sort of problems we will study in Sections \ref{sec:Si0nrelations} and
\ref{sec:Pi0nrelations}.

\section{Component-wise reductions in Descriptive Set Theory}

Component-wise reductions are in fact the standard, rather than a novel, concept
in descriptive set theory, and the modern interest in them in complexity and
computability theory is likely inspired by these earlier results.

The notion of Borel reducibility is introduced in Friedman and Stanley
(\cite{Friedman1989}),
Definition 1 and 2. Denoted by $\le_B$, if $A$ and $B$ are classes of structures
then $A\le_B B$ if and only if there exists a Borel function $f$ such that for
all $x,y\in A$ we have $f(x),f(y)\in B$ and:
\begin{equation}
x\cong y\iff f(x)\cong f(y).
\end{equation}
As expected, we say $B$ is Borel-complete if $B$ is Borel and for all Borel $A$,
$A\le_B B$.

The classes of trees, linear orders and groups are shown to be Borel-complete as
well as fields of a prime or zero characteristic. Examples of non-complete
classes are finitely branching trees, Abelian torsion groups and fields of a
prime or zero characteristic degree with a finite transcendence rank. It is
interesting that restricted versions of these natural mathematical structures
display similar behaviour in our framework also: we will present a result for
subgroups of $\QQ$ in Section~\ref{sec:Sigma03} and for polynomial time trees in
Section~\ref{sec:Pi01}.

The only equivalence relation studied by \cite{Friedman1989} is thus
isomorphism, however the framework was readily extended to a more general
setting. Much of these results are covered by Gao (\cite{Gao09}).

Two of the more famous results arising out of the fields are Silver's Dichotomy
(\cite{Silver1980}) and the Glimm-Effros dichotomy (Harrington, Kechris and
Louveau \cite{Harrington1990}). Silver's dichotomy states that a $\Pi^1_1$
equivalence relation's position in the degree hierarchy is either at most
equality on $\NN$ or at least equality over $2^\NN$. Another consequence of
\cite{Silver1980} is that equality over $2^\NN$ is the simplest Borel
equivalence relation with uncountably many equivalence classes. The Glimm-Effros
dichotomy states that there is no Borel equivalence relation between equality
over $2^\NN$ and $E_0$, where $E_0$ is the relation of almost equality, i.e.\
$(A,B)\in E_0$ if the symmetric difference of $A$ and $B$ is finite. Part of the
work in the computable case has been motivated by the desire to find similar
results.

\section{Component-wise reductions in Computability Theory} 

Given the breadth of the field, it is to be expected that the notion of
reducibility means different things to different authors. We will briefly
examine some alternate definitions before turning to those which correspond to
our own.

Calvert, Cummins, Knight and Miller (\cite{Calvert2004}) consider the degree
structure of classes of finite structures
under \emph{enumeration reducibility}, denoted by $\le_c$. If $A$ and $B$ are
two classes of finite structures, then $A\le_c B$ means there is a partial
computable function $\varphi$ such that if $a,a'\in A$ and $a\cong a'$ then
$\varphi(a)\cong \varphi(a')$. In our framework such an enquiry would have been
trivial: all finite structures are computable, so the only distinction between
relations on such would be the number of equivalence classes. The key difference
between the approach of \cite{Calvert2004} and the one considered in this work
is that in our study we are concerned with relations on the natural
numbers, i.e. subsets of $\NN\times\NN$, whereas in \cite{Calvert2004}
structures are encoded as subsets of $\NN$, so the isomorphism relation studied
there is in fact a subset of $2^\NN\times2^\NN$.

The resulting framework imposes a further structural constraint on reducibility.
Their encoding ensures that if $a\subseteq b$, then the structure encoded by $a$
is a substructure of that encoded by $b$. In Proposition 1.1, they show that
this induces that if $A\le_c B$ via $f$, then $a\subseteq b$ implies
$\varphi(a)\subseteq \varphi(b)$. In other words, reducibility must respect
substructures as well as isomorphism.

The resulting degree structure is rich. There are uncountably many classes
incomparable via $\le_c$ (Proposition 4.1), as well as countably infinite chains
(Proposition 4.5). A maximal element of the ordering is that of finite graphs,
and an open problem left by the authors is whether the class of finite graphs is
reducible to the class of finite linear orders. This parallels the Complexity
Theory results of \cite{Buss2011}, where the authors likewise find the class of
graphs to be maximal under strong isomorphism reducibility, but are unable to
determine whether the class of linear orders with a unary relation lies strictly
below graphs or not. The notion of computable reducibility on structures encoded
in $2^\NN$ is further explored
by Calvert and Knight (\cite{Calvert2006}). Knight, Miller and Vanden Boom
(\cite{Knight2007}) consider the same setting but
where the reducibility function is only required to be Turing.

Gao and Gerdes (\cite{Gao2001}) use the same notion of reducibility as we do,
labelling it
$m$-reducibility by analogy with classical recursion theory. They demonstrate
that the degree structure of equivalence relations has an initial segment of
type $\omega+1$. Namely, if $Id(n)$ is taken to mean equality modulo $n$ and
$Id$ equality over the the natural numbers, we have the following tower:
\begin{equation}
Id(1)\lneq Id(2)\lneq\dots\lneq Id.
\end{equation}
Every computable equivalence relation with $n$ equivalence classes is
bi-reducible with $Id(n)$, and those with infinitely many equivalence classes
are bi-reducible with $Id$ (Proposition 3.4). Together this gives a complete
characterisation of the computable equivalence relations: precisely those which
have a computable invariant function.

The authors further show that it is possible to embed the 1-degrees
into the degrees induced by component-wise reducibility. 
Given an r.e.\ set $A$, the authors define $R_A=\{(x,y)\ |\ x=y\vee x,y\in A\}$
and show that $A\le_1 B$ implies $R_A\le R_B$. This result is
strengthened by Coskey, Hamkins and Miller (\cite{Coskey2012}) to show that if
$A,B$ are non-computable then
$R_A\le R_B$ implies $A\le_1 B$. In other words, even though $\le$ does not
automatically imply $\le_1$, the degree structure of equivalence relations
nevertheless contains a copy of the 1-degrees.

It should be noted that while, as we have seen, the number of equivalence
classes is crucial to the complexity of a computable equivalence relation, the
size of these classes does not seem to matter much. This pattern persists
throughout the arithmetical hierarchy, and an
equivalence relation does not need to have an overly complex equivalence class
structure to rank highly with respect to $\le$. Fokina and Friedman
(\cite{Fokina2012:2}) show that it
is in fact possible to construct a properly $\Sigma^1_1$ equivalence relation,
vastly more complex than anything we deal with here, which has equivalence
classes of size one or two only (Claim 5.1).

Much of \cite{Coskey2012} focuses on the parallels of Borel
reducibility in a computable setting. They show that a version of Silver's
Dichotomy fails: equality of of r.e.\ sets does not appear to have any special
status in the computable world. There are relations strictly below and above it,
even within $\PI 2$. In fact, there exist infinite chains and arbitrarily large
finite antichains of equivalence relations lying below equality of r.e.\ sets
(Corollary 4.14). The authors leave the question as to whether a parallel of the
Glimm-Effros dichotomy exists open.

Most similar in spirit to our enquiry is that of Fokina, Friedman and Nies
(\cite{Fokina2012}). Their study is of equivalence relations complete for $\SI
3$. Their main result (Theorem 1) demonstrates that 1-equivalence on r.e.\ sets
is complete, which they later use to demonstrate that computable isomorphism of
computable predecessor trees, Boolean algebras and metric spaces is $\SI
3$-complete. This is a
technique we will use in Sections \ref{sec:Sigma03} and \ref{sec:Pi01}:
establish the completeness of a relation that may only seem meaningful to a
computer scientist, then use that result to show the completeness of more
natural
relations in mathematics.

\section{Component-wise reductions in Complexity Theory}\label{sec:lit_complex}

Polynomial time component-wise reductions seem to be first introduced by Fortnow
and Grochow
(\cite{Fortnow2011}, Definition 4.13). As the kernel of a polynomial time
function $f$ is
the set of $(x,y)$ satisfying $f(x)=f(y)$, \cite{Fortnow2011} define a
\emph{kernel
reduction} from $E$ to $S$ to be a polynomial time function satisfying:
\begin{equation}
Exy\iff Sf(x)f(y).
\end{equation}
In other words, $x,y,$ lie in the kernel of $f$ modulo $S$.

The interest in this notion of reduction was due to the fact that an equivalence
relation kernel-reduces to equality if and only if it has a polynomial time
computable invariant. The paper
leaves as an
open question whether kernel reductions are actually distinct from standard
polynomial time reductions, and if so whether P has a maximum element under
kernel reductions.

The issue was picked up on by Buss, Chen, Flum, Friedman and M\"uller
(\cite{Buss2011}) in their investigation of \emph{strong
isomorphism
reductions}, $\le_{iso}$, and \emph{strong equivalence reductions}, $\le_{eq}$,
depending on whether the
underlying relation concerned the isomorphism of finite objects or a general
equivalence relation.

More specifically, they define a strong isomorphism reduction to be a polynomial
time $f$ between
classes of arbitrarily large structures closed under isomorphism, where a
structure is a finite tuple of relations over a finite domain. Then we can say
for classes $C,D$ that $C\le_{iso} D$ if and only if for all
$\mathcal{A},\mathcal{B}\in C$:
\begin{equation}
\mathcal{A}\cong\mathcal{B}\iff f(\mathcal{A})\cong f(\mathcal{B}).
\end{equation}

Strong equivalence reductions are defined more in line with the component-wise
reductions we defined in \eqref{eq:component_wise}. As such $\le_{eq}$ does not
range over classes, but rather equivalence relations. As expected, a polynomial
time $f$ is a reduction from $R$ to $E$, denoted $R\le_{eq}E$, if it satisfies:
\begin{equation}
(x,y)\in R\iff (f(x),f(y))\in E.
\end{equation}

Within the
paper the open question of \cite{Fortnow2011} regarding the distinction of
$\leq_m^P$ from a component-wise reduction was answered by demonstrating a rich
structure of strong isomorphism degrees within P. The authors demonstrate that
the countable, atomless Boolean algebra is embeddable into the degree structure
induced by $\le_{iso}$ on polynomial time classes of structures (Theorem 5.1).

There thus exist infinite chains of equivalence relations under strong
isomorphism reductions despite being equivalent under $\le_m^P$; one is
therefore able to achieve an infinitely finer gradation of degrees with
component-wise reductions in complexity theory.

As we have seen in the previous section, for computable equivalence relations
the number of equivalence classes determines the degree structure.
\cite{Buss2011} study a parallel notion as \emph{potential reducibility}: if
$A(n)$
and $B(n)$ represent the number of isomorphism types of structures of size at
most $
n$, the authors observe that if $A$ strong isomorphism reduces to $B$, it must
be the case that for some polynomial $p$, $A(n)\le B(p(n))$ for all $n\in\NN$.
When the latter condition holds $A$ is said to potentially reduce to $B$.

The authors were unable to determine whether potential reducibility differs from
their other notions in the absence of further complexity theoretic assumptions.
However they do demonstrate that if potential reducibility and strong
isomorphism
reducibility is distinct then P$\neq\#$P. For
strong equivalence reducibility, being distinct from potential reducibility
would imply P$\neq$NP.

A curious distinction between the case with complexity and computability
reductions is that while the existence of complete equivalence relations, at
least for $\SI n$, was never in question, the authors of
\cite{Buss2011} were unable to answer the
second question of \cite{Fortnow2011} definitively. They did, however, manage to
derive a necessary and sufficient condition: any of the classes of polynomial
time, NP, or CoNP equivalence relations admit a
complete element if and only if there exists an effective enumeration of that
class.

The ``if" direction clearly carries over to the computable case: this is
precisely the argument that guarantees the existence of $\SI n$ complete
equivalence relations. The second direction, however, relies on the fact that
there exists an effective enumeration of polynomial time functions, and as such
given a maximum element the class of polynomial time equivalence relations it
would be possible to enumerate all
elements
of that class via the preimage of the reduction functions. We cannot
replicate such a technique in the computable case.
However, it does not appear that we need to: in Section~\ref{sec:Pi0nrelations}
we will demonstrate the existence of a $\PI
1$-complete equivalence relation, whereas we do not know of any effective
enumeration of that class.

\chapter{The case of \texorpdfstring{$\mathbf{\SI n}$}{Sigma 0n}}\label{sec:Si0nrelations}

While the existence of $\SI n$-complete equivalence relations and preorders
follow immediately from the transitive closure of a $\SI n$ relation being $\SI
n$, natural examples of such relations are nevertheless interesting as
component-wise reducibility induces a finer granularity on the arithmetical
hierarchy than $\le_m$. In this section we give three examples of preorders and
associated equivalence relations from logic, complexity theory and algebra
complete for $\SI 1,\SI 2$ and $\SI 3$ respectively.

\section{\texorpdfstring{$\mathbf{\SI 1}$}{Sigma 01} and Logical Equivalence}

The reader is reminded that the following problem is $\SI 1$-complete:
\begin{equation}
\{\varphi\ |\ \exists\Phi:\ \Phi\text{ is a FOL proof of }\varphi \}.
\end{equation}
By FOL we mean first order logic with a full signature; that is, with
predicate symbols of any arity\footnote{Of course, the full strength of first
order logic is not required. The papers of Church and Turing
(\cite{Church1936},\cite{Turing1937}) established this result only needing
predicates of arity at most 2. However as having predicates of higher arity is
convenient for our construction, we do not place ourselves under any
restrictions here.}. This yields a natural $\SI 1$ equivalence
relation, logical equivalence:
\begin{equation}
\{(\varphi,\psi)\ |\ \exists\Phi:\ \Phi\text{ is a FOL proof of }\varphi\lra\psi
\}.
\end{equation}
In a sense this result is not new. Montagna and Sorbi (\cite{Montagna1985}) have
shown that implication in any consistent, finitely axiomatisable r.e.\ extension
of Peano arithmetic is
complete for $\SI 1$-preorders. In an earlier result Pour-El and Kripke
(\cite{Pour-El1967}) have shown that finitely axiomatisable theories
containing a sufficient portion of Peano Arithmetic are complete with respect to
the class of such theories under a deduction-preserving mapping. As the
conjunction of any finite set of axioms can be thrown into the antecedent of a
first order sentence, both of these frameworks can extend to our result. What we
present below, however, is a direct proof,
and the only fragment of arithmetic we require is that which is already
available in first order logic.

\begin{theorem}
Logical implication is complete for $\SI 1$ preorders. That is, for every $\SI
1$ preorder $P$ and every $x,y$ there exists a computable $f$ mapping integers
to FOL formulae such that:
\begin{equation}
Pxy\iff \vdash_{FOL}f(x)\ria f(y).
\end{equation}
\end{theorem}
\begin{proof}
Observe that $P\in\SI 1$ implies that for a given $x$ the set $\{y\ |\ Pxy\}$
is uniformly r.e., meaning there exists a machine which enumerates it. Our proof
relies on representing the computation of such a machine in first order logic.
We define a non-deterministic $M$ that on input $x$ begins to enumerate all $y$
above it. Every time a new $y$ is found, the machine makes a non-deterministic
choice. One branch will continue with the described computation, whereas the
other will ``restart" on input $y$. That is: clear the tape, write $y$ on it,
move to the first cell and enter the initial state. Observe that if $Pxy$, some
branch of $M$ on input $x$ will eventually enter the initial configuration of
$M$ on input $y$.

The configurations of $M$ will correspond to formulae in the logic, and our
construction will ensure that if $\varphi$ and $\psi$ are representations of
configurations, and $M$ can make a transition from $\varphi$ to $\psi$, then
$\varphi\ria\psi$ will be provable in the logic. Therefore our reduction will be
given by an $f$ that maps $x$ to the the initial configuration of $M$ on input
$x$. Readers familiar with such constructions can safely skip the rest of the
proof: it contains nothing but technicalities.

Let us be more specific about our machine. We require a non-deterministic $M$
with two tapes: a working tape and a printing tape.
For simplicity, $M$ is working on a unary alphabet. The initial configuration of
$M$ on input $x$ then consists of the first $x$ cells of the working tape
marked, both heads at the leftmost edge of their tapes and the machine in the
initial state, $q_0$. When run on a tape with the first $x$ cells
marked the machine will print all $y$ with $Pxy$. We take printing $y$ to mean
that it will mark
the first $y$ cells of the print tape, then enter a special state $q_{pr}$.

After exiting $q_{pr}$ it enters another special state $q_{nd}$ and then makes a
non-deterministic branch. One
branch continues searching for elements above $x$ in $P$, the other clears both
tapes, marks the first $y$ cells of the working tape, moves the heads to the
start of the tapes, enters a special state $q_{re}$, then enters the initial
state $q_0$. Intuitively, one branch continues
with the computation, the other
restarts computation on input $y$. For the sake of convenience, we
assume the machine has the ability to keep the head still and not write or read
anything during a transition, which it only ever uses in states $q_{pr},q_{nd}$
and $q_{re}$. We also require that the \emph{cleaning} states, i.e.\ those used
between $q_{nd}$ and $q_{re}$, are not used anywhere else in the computation,
and the special states $q_{pr},q_{nd}$ and $q_{re}$ are not used anywhere except
in the circumstances outlined above.

The idea behind the construction is to have $f$ map $x$ and $y$ to the initial
configurations of $M$ on input $x$ and $y$ respectively, and then represent the
computation of $M$
in first order logic. If $Pxy$, then some computation branch of $M$ on input $x$
will eventually enter the initial configuration of $M$ on input $y$ and as such
it will be possible to derive $f(y)$ from $f(x)$.

As a convention, we will treat the first cell on the tape as 1 rather than 0.

We require the following predicate symbols in our logic. For the reader's
convenience we note the intended interpretation in parentheses, but of course as
far as syntax is concerned these are just arbitrary predicates.
\begin{itemize}
\item
$C_0$: 4-ary predicate ($C_0w_ms_td_bu_x$ means cell $m$ of the working
tape is marked at time $t$ with initial machine input $x$, and $d_b$ is a
counter used to track the branch of the computation).
\item
$C_1$: 4-ary predicate (As $C_0$, but for the output tape).
\item
$Q_i$ 4-ary predicate for all machine states $q_i$ ($Q_i s_td_bu_xn_z$ means the
machine is
in state $q_i$ at time $t$ in branch $b$ with initial input $x$. The variable
$n_z$ is $y+1$ if the machine is in the process of cleaning the tape to restart
on $y$, otherwise it is 0).
\item
$H_0$: 4-ary predicate ($H_0w_ms_td_bu_x$ means working head is over cell $m$
at time
$t$ with initial machine input $x$ on branch $b$ of the computation).
\item
$H_1$: 4-ary predicate (Same as $H_0$ but for the output head).
\item
$P$: 4-ary predicate ($Pw_ys_tu_x$ means the machine printed $y$ at time $t$
with
initial input $x$).
\item
$S$: Binary predicate ($Sxy$ means $y=x+1$).
\item
$Z$: Unary predicate ($Zx$ means $x=0$).
\item
$L$: Binary predicate ($Lxy$ means $x\le y$).
\item
$E$: Binary predicate ($Exy$ means $x=y$).
\end{itemize}
A few notes on our use of the variables: the intended interpretation of our
domain is the natural numbers, and so we use subscripts above to signal the
number a specific variable is meant to represent. As our variables will
generally be bounded by quantifiers, we can assume without loss of generality
that they are always renamed to conform with the conventions above. E.g.\
$s_{i+1}$ represents the time step following $s_i$, $w_m$ is the $m$th cell and
so on. We do however wish to remind the reader that this is merely for
notational intuition: these are arbitrary variables, and arithmetic is not a
part
of our logic.

We require the branch counters on certain predicates because as our machine is
non-deterministic, its configuration is not uniquely determined by the clock.
The head, for instance, can be in several positions at a given time and if we
are to prevent the head in a parallel computation branch from interfering, we
need to keep track of where in the computation tree each is located. This is
purely a technicality, the only places in which we will touch the branch counter
is in state $q_{nd}$, where the choice happens, and in state $q_{re}$, where the
machine restarts.

The variable $n_z$ is carried by the state predicate purely so that we can
``remember" the value we are supposed to restart the machine on. We alter it
only after the print state, and when the machine is restarted.

It is necessary
to record the initial input via $u_x$ because a single predicate captures not
the entire machine configuration, but just a component of it: it is entirely
possible for two different configurations to exist at the same time and
same branch number of the computation, but with different initial input. We do
not wish these to interfere with each other. The only time we change $u_x$ is
when we reset the machine.

Next we define some formulae to capture the behaviour of $M$ and the
characteristics of the natural numbers. As is usual in logic we cannot guarantee
the uniqueness of zero and successor, so we have to settle for uniqueness modulo
$E$.

\begin{enumerate}
\item
$\forall x\ Lxx$ (reflexivity of $L$).
\item
$\forall xyz\ Lxy\wedge Lyz\ria Lxz$ (transitivity of $L$).
\item
$\forall xy\ Lxy\vee Lyx$ (totality of $L$).
\item
$\forall xy\ Exy$ (reflexivity of $E$).
\item
$\forall xyz\ Exy\wedge Eyz\ria Exz$ (transitivity of $E$).
\item
$\forall xy\ Exy\ria Eyx$ (symmetry of $E$).
\item
$\forall xy\ Zx\ria Lxy$ (zero is least element).
\item
$\forall xy\ Sxy\ria Lxy$ ($L$ respects successor).
\item
$\forall xy\ Zx\ria\neg Syx$ (zero has no predecessor).
\item
$\forall x\exists y\ Sxy$ (every number has a successor).
\item
$\forall x\exists yz\ Sxy\wedge Sxz\ria Eyz$ (uniqueness of successor modulo
$E$).
\item
$\forall x\exists y\ \neg Zx\ria Syx$ (every non-zero number has a predecessor).
\end{enumerate}
As well as an indiscernibility of identicals clause, stating that $Exy$ implies
that the predicates we defined above are true for $x$ if and only if they are
true for $y$. Note that this is a valid first order formulae because we have
defined only finitely many predicate symbols. The conjunction of the above we
will label $NAT$.

We next turn to the machine
behaviour. Recall that every transition of $M$ can be represented as a sextuple
$(T,q_i,R,W,M,q_j)$, where $q_i$ is the machine
state before the transition,
$R\in\{0,1\}$ represents whether or not the cell under the head on tape $T$ is
marked,
$W\in\{0,1\}$ the action of the machine on the cell under the head on tape $T$,
$M\in\{Le,Ri\}$
the direction the head moves on tape $T$ and $q_j$ the end state.

\begin{enumerate}
\item\label{enum:trans_one}
For every transition of the form $(T,q_i,R,W,M,q_j)$, where
$q_i\notin\{q_{pr},q_{nd},q_{re}\}$:
$$\forall w_ms_td_bu_xn_z\exists s_{t+1}\ (Ss_ts_{t+1}\wedge (ANT\ria
(Q_js_{t+1}d_bu_xn_z\wedge TAPE\wedge HEAD))).$$

$ANT$ states the machine read $R$ in state $q_i$ on tape $T$. Note that
whether or not the last clause is present is contingent on the value of $R$:
$$Q_is_td_bu_xn_z\wedge H_Tw_ms_td_bu_x\wedge C_Tw_ms_td_bu_x[\text{if
}R=1].$$

$TAPE$ states that the cells not under the head on tape $T$ remain
unchanged:
$$\forall w_kk(\neg H_Tw_ks_td_bu_x\ria$$
$$((C_Tw_ks_td_bu_x\ria
C_Tw_ks_{t+1}d_bu_x)\wedge(C_{1-T}w_ks_td_bu_x\ria
C_{1-T}w_ks_{t+1}d_bu_x))).$$

$HEAD$ states that the $T$-head writes $W$ and moves $M$, while the other
head remains where it is. Note that the formulae will vary depending on the
values of $W,R$ and $M$:

$C_Tw_ms_{t+1}d_bu_x\ [\text{if }W=1]\ \wedge\\ (\forall i\ Zi\ria\\((Siw_m\ria
H_Tw_ms_{t+1}d_bu_x)\wedge(\neg
Siw_m\ria(\exists w_{m-1}\ Sw_{m-1}w_m\wedge H_Tw_{m-1}s_td_bu_x))))
\\\relax [\text{if }M=Le]\ \wedge\\ (\exists w_{m+1}\
Sw_mw_{m+1}\wedge H_Tw_{m+1}s_{t+1}d_bu_x)\ [\text{if }M=Ri]\ \wedge\\
(\forall w_ms_td_bu_x\ H_{1-T}w_ms_td_bu_x\ria
H_{1-T}w_ms_{t+1}d_bu_x).$

\item\label{enum:trans_two}
For every transition with $q_i=q_{pr}$, we remind the reader that $M$ prints the
contents of the output tape and transitions to $Q_{nd}$ without moving the heads
or writing anything. So we have:
$$\forall w_ms_td_bu_xn_z\exists s_{t+1}\ (Ss_ts_{t+1}\wedge$$
$$(Q_{pr}s_td_bu_xn_z\ria (Q_{nd}s_{t+1}d_bu_xn_z\wedge TAPE\wedge
HEAD\wedge PRINT))).$$

$TAPE$ preserves the tape contents:
$$(C_0w_ms_td_bu_x\ria C_0w_ms_{t+1}d_bu_x)\wedge(C_1w_ms_td_bu_x\ria
C_1w_ms_{t+1}d_bu_x).$$

$HEAD$ preserves the head locations:
$$(H_0w_ms_td_bu_x\ria H_0w_ms_{t+1}d_bu_x)\wedge(H_1w_ms_td_bu_x\ria
H_1w_ms_{t+1}d_bu_x).$$

$PRINT$ prints the contents of the output tape:
$$\forall w_yw_m (C_1w_ys_td_bu_x\ria (C_1w_ms_td_bu_x\ria Lw_mw_y))\ria
Pw_ys_{t+1}u_x.$$

\item\label{enum:trans_three}
For every transition with $q_i=q_{nd}$ and $q_j$ a cleaning state, we remind the
reader that $M$ makes a transition to $q_j$ without doing anything else. All we
need to do is preserve the tape contents, update the branch variable and set
$n_z$ to remember the printed value. Note
that this will induce that all cleaning steps have an odd branch counter:
$$\forall w_ms_td_bu_xn_zw_y\exists s_{t+1}d_{b+1}\ (Ss_ts_{t+1}\wedge
Sd_bd_{b+1}\wedge$$
$$(Pw_ys_tu_x\wedge Q_{nd}s_td_bu_xn_z\ria (Q_{j}s_{t+1}d_{b+1}u_xw_y\wedge
TAPE\wedge
HEAD))).$$

$TAPE$ preserves the tape contents:
$$(C_0w_ms_td_bu_x\ria
C_0w_ms_{t+1}d_{b+1}u_x)\wedge(C_1w_ms_{t}d_bu_x\ria
C_1w_ms_{t+1}d_{b+1}u_x).$$

$HEAD$ preserves the head locations:
$$(H_0w_ms_{t}d_bu_x\ria
H_0w_ms_{t+1}d_{b+1}u_x)\wedge(H_1w_ms_{t}d_bu_x\ria
H_1w_ms_{t+1}d_{b+1}u_x).$$

\item\label{enum:trans_four}
For every transition with $q_i=q_{nd}$ and $q_j$ a non-cleaning state, we remind
the
reader that $M$ makes a transition to $q_j$ without doing anything else. All we
need to do is preserve the tape contents and update the branch variable. Note
that this induces that all the non-cleaning steps have an even branch counter:
$$\forall w_ms_td_bu_xn_z\exists s_{t+1}d_{b+1}d_{b+2}\ Ss_ts_{t+1}\wedge
Sd_bd_{b+1}\wedge Sd_{b+1}d_{b+2}\wedge$$
$$(Q_{nd}s_td_bu_xn_z\ria
(Q_{j}s_{t+1}d_{b+2}u_xn_z\wedge TAPE\wedge
HEAD)).$$

$TAPE$ preserves the tape contents:
$$(C_0w_ms_td_bu_x\ria
C_0w_ms_{t+1}d_{b+2}u_x)\wedge(C_1w_ms_td_bu_x\ria
C_1w_ms_{t+1}d_{b+2}u_x).$$

$HEAD$ preserves the head locations:
$$(H_0w_ms_td_bu_x\ria
H_0w_ms_{t+1}d_{b+2}u_x)\wedge(H_1w_ms_td_bu_x\ria
H_1w_ms_{t+1}d_{b+2}u_x).$$

\item\label{enum:trans_five}
For every transition with $q_i=q_{re}$ recall that $q_j=q_0$ and all other
parameters are ignored. The heads should already be at the start of the tapes,
so all we need to do is preserve the tape contents and reset the counters. We
then have:
$$\forall w_ms_td_bu_xn_zi\exists u_yw_1\ (Zi\wedge Q_{re}s_td_bu_xn_z\ria$$
$$(Su_yn_z\wedge Siw_1\wedge Q_0iiu_yi\wedge TAPE\wedge
H_0w_1iiu_y\wedge H_1w_1iiu_y)).$$

Since the output tape should be empty, $TAPE$ is:
$$C_0w_ms_td_bu_x\ria C_0w_miiu_y.$$
\end{enumerate}

The conjunction of the above we will label $TRA$.

We will use $CONF(x)$ as shorthand for the formula representing the initial
configuration of $M$ on input $x$. That is:
$$\exists w_1\dots w_xu_1\dots u_x \
Z0\wedge$$
$$\bigwedge_{i=0}^{x-1}(Sw_iw_{i+1}\wedge Su_iu_{i+1})\wedge
Q_000u_x0\wedge H_0w_100u_x\wedge H_1w_100u_x\wedge\bigwedge_{i=1}^{x}
C_0w_i00u_x.$$

We now have all the components necessary to define $f$. Namely:
$$f(x)=NAT\wedge
TRA\wedge CONF(x).$$
It remains to show that $Pxy$ is equivalent to $f(x)\ria
f(y)$.

First, suppose that $\neg Pxy$. Recall that the completeness of first order
logic implies a formula is provable if and only if it is true under all
interpretations. In particular, it must also be true under the interpretation we
have been using throughout the construction: that of the simulation of $M$. As
$\neg Pxy$ implies that
$M$ on input $x$ will never reach the initial configuration of $M$ on input $y$,
$f(x)\ria f(y)$ is false under some interpretation and thus
cannot be provable in the logic.

Now suppose $Pxy$. Observe that this implies there exists a computation history
of $M$ starting with the initial configuration of $M$ on $x$ and leading to the
initial configuration of $M$ on $y$. Let $HIST(x,t)$ be the formula representing
the machine configuration at time $t$ of this history. We will show that for
$0\le t\le r-1$,
$\bigwedge_{i\leq t} HIST(x,i)\ria HIST(x,t+1)$.
This established, by the principle of strong induction we will have $f(x)\ria
f(y)$.

Note that there is no generality
lost in assuming that $r-1$ is the only step in the history where the state is
$q_{re}$. By the transitvity of $P$, whatever elements $M$ may reach after
restarting on another element, it will eventually reach by taking the
non-cleaning branches until the desired number is printed.

Before we proceed, we must be explicit about what we mean by $HIST(x,t)$
representing the machine configuration. Let $p$ be the farthest cell to the
right on either tape either marked or with a head above it, $b$ the current
branch counter and $n$ either 0 or $y+1$ depending on whether the machine has
entered the cleaning stage. Let $t\neq r$.

$$HIST(x,t)=\exists w_1\dots w_ps_1\dots s_tu_1\dots u_xd_1\dots d_bn_1\dots
n_z$$
$$NAT\wedge TRA\wedge Z0\wedge
Succ(x,t)\wedge State(x,t)\wedge
Heads(x,t)\wedge Tape(x,t).$$
$Succ(x,t)$ associates all active variables with the numbers they represent: 
$$Succ(x,t)=S0w_1\wedge S0s_1\wedge S0u_1\wedge S0d_1\wedge S0n_1\wedge$$
$$\bigwedge_{i=1}^{p-1}Sw_iw_{i+1}\wedge\bigwedge_{i=1}^{t-1}Ss_is_{i+1}
\wedge\bigwedge_{i=1}^{x-1}Su_iu_{i+1}\wedge\bigwedge_{i=1}^{z-1}Sn_in_{i+1}.$$
$State(x,t)$
is the current state, $Heads(x,t)$ is the positions of the heads and $Tape(x,t)$
is the contents of the tape. Note that this is consistent with our requirement
that $HIST(x,0)=f(x)$. $HIST(x,r)$ we let equal to $f(y)$ as mentioned above. We
will show that the implication holds for every component individually.

$NAT$, $TRA$ and $Z0$ are constant throughout, so clearly $NAT\wedge TRA\wedge
Z0\ria NAT\wedge
TRA\wedge Z0$ via the law of identity.

The $Succ(x,t)$ clause is also simple. To show $Succ(x,t)$ we need only show the
existence of sufficiently many successor elements, but $NAT$ allows us to always
derive the existence of a new variable that is the successor to some existing
variable. As such for all $t$, $NAT\wedge Z0\ria Succ(x,t+1)$ simply by
introducing the correct number of successors and renaming them as needed.

If the machine is in state $i$ at step $t$, then $State(x,t)$ is simply
$Q_itbw_xz$. It is easy to verify that if $t<r-1$ then we can derive
$Q_j(t+1)bw_xz$ for the appropriate $j$ via one of
\ref{enum:trans_one}-\ref{enum:trans_four}, and hence $State(x,t)\ria
State(x,t+1)$. When $t=r-1$, however, while \ref{enum:trans_five} says that we
can derive
$Q_000u_y0$, we must be careful not to be deceived by notation: $u_y$ in this
context is a simply another name for $n_{z-1}$ ($n_z$ being such that
$Q_{re}s_{r-1}d_bu_xn_z$ holds at time $r-1$), whereas in $State(x,s)=Q_000u_y0$
it is specifically the $y$th successor of zero. To prove the implication holds
we
must therefore show that $En_{z-1}u_y$ holds at time $r$; that is, is derivable
from $\bigwedge_{t<r} HIST(x,t)$.

Let $t'$ be the time in which the machine prints $y$. Observe that the machine
is in state $q_{pr}$, so \ref{enum:trans_two} applies. From the $PRINT$
statement we can derive $Pw_ys_{t'+1}u_x$. In the context of $PRINT$, $w_x$ is
the rightmost cell on the output tape. Since at time $t'$ the output tape has
the first $y$ cells marked, $w_y$ is the $y$th successor of zero. As $NAT$
accounts for the uniqueness of successor, $(HIST(x,t')\wedge NAT)\ria Ew_yu_y$.
In state $t'+1$, the machine enters $q_{nd}$ where from \ref{enum:trans_three}
we get $Q_js_{t'+2}d_{b+1}u_xw_y$. Note that from this point on the last
argument of the state predicate is $w_y$ where $Ew_yu_y$, we can rename
variables to obtain $Q_{re}s_{r-1}d_bu_xn_z$ and $En_zu_y$ as required. This
gives us $\bigwedge_{i\le t} HIST(x,i)\ria State(x,t)$.

For the $Heads(x,t)$ clause we must preserve the position of the heads that do
not move, and move the heads that do. In states $q_{pr},q_{nd}$ and $q_{re}$
both heads are stationary, and this is accounted for by the $HEAD$ clause of
\ref{enum:trans_two}-\ref{enum:trans_five}. If a head does move, we are in the
case covered by \ref{enum:trans_one} and it is easy to verify that the $HEAD$
clause there does all that is required: move the active head left or right,
keeping it still if it tries to move off the edge of the tape, and preserving
the position of the other head. For the transition to $Heads(x,r)$ we must also
verify that $u_y$ is indeed the $y$th successor of zero, which we do by the same
argument we used for $State(x,t)$, giving us $\bigwedge_{i\le t} HIST(x,i)\ria
Heads(x,t)$.

For the $Tape(x,t)$ clause we must make sure the marked tape cells are preserved
and if the head marks a new cell it is accounted for. Note that if the machine
is in state $q_{pr},q_{nd}$ or $q_{re}$, the head takes no action and the
preservation of the tape follows immediately from the $TAPE$ clause of
\ref{enum:trans_two}-\ref{enum:trans_five}. If the machine is in a different
state the head must either mark or clear the cell it is over. Note that the
$HEAD$ clause in \ref{enum:trans_one} provides for exactly this, and the $TAPE$
clause preserves the rest of the tape. It follows that $Tape(x,t)\ria
Tape(x,t+1)$ for $t<r-1$. For the last implication we need to assure ourselves
that $u_y$ is the $y$th successor of zero, which we do by the same argument as
before, giving us $\bigwedge_{i\le t} HIST(x,i)\ria Tape(x,t)$.

From all of the above it follows that for $t<r$, $\bigwedge_{i\leq t}
HIST(x,i)\ria HIST(x,t+1)$. By strong induction it follows that if $Pxy$ then we
can derive $f(y)$ from $f(x)$, completing the theorem.
\end{proof}

\begin{corollary}
Logical equivalence is complete for $\SI 1$ equivalence relations.
\end{corollary}
\begin{proof}
Follows from Proposition~\ref{prop:preorders}, as equivalence is the symmetric
fragment of implication.
\end{proof}

\section{\texorpdfstring{$\mathbf{\SI 2}$}{Sigma 02} and Polynomial Time
Equivalence}

With two lead quantifiers, $\SI 2$ admits a wider array of candidates for our
study as it is now simple to speak of the behaviour of mappings from one
structure to another: we can use the lead existential quantifier to select a
function from a suitable class, and the universal to dictate the function's
behaviour on all input. We choose to consider a preorder and a
resulting equivalence relation
from complexity theory: that of polynomial time reducibility of EXPTIME sets. In
fact we do not need the full power of EXPTIME, treating it as DTIME($2^n$)
would suffice.

We first prove a lemma which establishes a connection between component-wise and
$m$-reducibility. Since when viewed as a set of pairs a $\SI n$ preorder is just
a $\SI n$ set, we can bootstrap classical completeness results to provide a
useful characterisation of the preorder in question. The technique is quite
general, and we make further use of it in Sections \ref{sec:Sigma03} and
\ref{sec:Pi0n}.

\begin{lemma}\label{lemma:Vxy_sequence}
For every $\SI 2$ preorder $\preceq$ there exists a sequence of non-empty r.e.\
sets $\{V_{xy}\ |\ x,y\in\NN\}$ such that $x\preceq y$ iff $V_{xy}$ is finite,
and given $x$
and $y$ we can retrieve the machine enumerating $V_{xy}$ in time linear in $\log
x,\log y$. That is, linear in the size of the input.
\end{lemma}
\begin{proof}
Recall that the following problem is $m$-complete for $\SI 2$ (see, for instance,
\cite{Soare1987}
Theorem IV:3.2):
$$\{e\ |\ W_e\text{ is finite}\}.$$

It follows that there exists a computable $f$ such that $(x,y)\in\preceq\iff
W_{f(x,y)}$ is finite. Such a $W_{f(x,y)}$ is then precisely the $V_{xy}$
required. Observe that there is no generality lost in assuming $V_{xy}$ is
non-empty: we
can merely define $V_{xy}$ to be $W_{f(x,y)}\cup\{0\}$, thereby guaranteeing
that it has at least one element. It remains to show that we can compute $f$ in
linear time.

As $\preceq$ is a $\SI 2$ set, there exists a recursive predicate $R$ such that
$x\preceq y\iff\exists v \forall u\ Rvuxy$. Let $M$ be the (always halting)
machine computing $R$ and $N$ the machine that on input $\la
x,y\ra$ runs $M$ on $\la v=0,u,x,y\ra$, for all $u$. If $M$ ever rejects,
$N$ prints $v$ and runs $M$ on $\la v+1,u,x,y\ra$ for all $u$, and so on.
Observe that if $x\preceq y$ then $N$ will only print finitely many values, as
eventually it will reach such a $v$ such that $Rvuxy$ holds for all $u$. We will
therefore use $N_{xy}$ to denote $N$ with input fixed to $\la x,y\ra$, $N_{xy}$
is
precisely the machine enumerating $V_{xy}$.

Let $|N|$ be the size of $N$. Observe that to obtain $N_{xy}$ from $N$, it is
sufficient to add $2|\la x,y\ra|$ states to $N$: $|\la x,y\ra|$ to print $\la
x,y\ra$ on the tape, $|\la x,y\ra|$ to move the head back to the first cell and
then enter the initial state of $N$. The size of $N_{xy}$ is then
$|N|+2|\la x,y\ra|$, so if the encoding is linear in $\log x,\log y$, then we
can obtain a
machine enumerating $V_{xy}$ in time linear in $\log x,\log y$.
\end{proof}

\begin{theorem}
Polynomial time reducibility of EXPTIME sets is complete for $\SI 2$ preorders.
That is, if $A_x$ is the $x$th exponential set, then for every $\SI 2$ preorder
$\preceq$ there exists a computable function $f$ satisfying:
$$x\preceq y\iff A_{f(x)}\le^P_mA_{f(y)}.$$
\end{theorem}
\begin{proof}
First, let us establish that this relation is in fact $\SI 2$. Note that
$A_i\le^P_m A_k$ if and only if there exists a polynomial time $t$
such that $x\in A_i\iff t(x)\in A_k$. Taking $t_e$ to be the $e$th polynomial
time function, the desired relation can be expressed as follows:
$$\{(i,k)\ |\ \exists e\forall x\  x\in A_i\lra t_e(x)\in A_k\}.$$
The matrix is computable: we can query $A_i$ in exponential time, compute $t_e$
in polynomial and verify whether the result is in $A_k$ in exponential. As a
result the relation is $\SI 2$.

Let $\preceq$ be a $\SI 2$ preorder, $\{V_{xy}\ |\ x,y\in\NN\}$ be a sequence
of non-empty r.e.\ sets as in Lemma~\ref{lemma:Vxy_sequence}, $\{M_e\ |\
e\in\NN\}$
be an enumeration of all polynomial time oracle
machines, and $p_e$ be the clock of the $e$th machine.
We will construct EXPTIME sets $\{A_x\ |\ x\in\NN\}$ consisting of strings over
the alphabet $\{0,1\}$ satisfying:
\begin{enumerate}
\item\label{enum:SI2finite}
$V_{xy}$ finite $\ria$ $A_x\leq^P_mA_y$.
\item\label{enum:SI2infinite}
$V_{xy}$ infinite $\ria$ $A_x\nleq^P_TA_y$. That is, $A_x\neq M_e(A_y)$ for any
$e$.
\end{enumerate}
Note that as $\le_m^P$ implies $\le_T^P$, this will be sufficient to establish
the theorem.

The idea behind the construction is to make the sets extremely sparse, so that
elements we add to the set later in the construction will be much too large to
affect the computation of any polynomial time reduction. To this end we define
the function $g$ by $g(0)=1$, $g(n+1)=2^{g(n)}$.
We also have need of a unary encoding of 4-tuples that encodes each 4-tuple by a
string of length $g(n)$ for some $n$. For convenience, if $(x,y,r,e)$ is encoded
by $0^m$ then we require that
$x,y,r,e\leq m$. Clearly for a function that grows as quickly as $g$ this will
be satisfied almost everywhere anyway, but insisting on the restriction reduces
the cases needed to consider in the proof.

To ensure \eqref{enum:SI2finite} the construction will make use of a suffix
table.
This table will associate a pair $(x,y)$ with a 4-tuple $(x,y,r,e)$ and a
polynomial $q$.
The sets $\{A_x\ |\ x\in\NN\}$ will be constructed by stages, and the meaning of
the suffix table is that if at stage $n$ of the construction $(x,y)$ is
associated with $(x,y,r,e)$ and
$q$ then for all $w$ with $|w|\geq g(n)$, $w\in A_x$ at stage $n$ iff $w1^{\la
x,y,r,e\ra}0^{q(|w|)}\in A_y$.
The table thus defines a polynomial reduction of $A_x$ at stage $n$ into $A_y$:
if $w$ is long enough we can apply the suffix, if it is too short we can verify
whether
it belongs to $A_x$ directly in constant time.
If after some stage of the construction the $(x,y)$ entry is never updated, we
will have a reduction at all stages, and hence a reduction of the entire set.

To ensure \eqref{enum:SI2infinite} we will use $P_{xye}$ to denote the
requirement
that $A_x\neq M_e(A_y)$. This will be done by letting $A_x(w)=1-M_e(A_y;w)$ for
some $w$.
The construction will ensure that if $V_{xy}$ is large enough then $P_{xye}$
will be satisfied.

Regarding terminology, note that we distinguish between $P_{xye}$ being
satisfied and being declared satisfied. When we say that $P_{xye}$ has been
declared satisfied, we mean that the stage of the construction below where the
declaration is made has been reached. When we say that $P_{xye}$ is satisfied we
mean that it is true: that is, $A_x\neq M_e(A_y)$ whether we have declared
$P_{xye}$ to be satisfied or not.

Once satisfied, $P_{xye}$ is never injured. The idea is then that if $V_{xy}$ is
finite, we will have a reduction via the suffix table. If $V_{xy}$ grows
indefinitely, it will eventually result in every $P_{xye}$ being satisfied,
meaning no machine can be a reduction.

At stage $n$ of the construction let $m=g(n)$, and as $g(n)$ is the encoding of
a 4-tuple, $m=\la x,y,r,e\ra$. The purpose of the stage is to determine whether
$w=0^m$ is in $A_x$.

First verify that the $e$th element enters $V_{xy}$ at step $r$.
If not, do nothing at this stage.

Update the $(x,y)$ entry of the suffix table with $(x,y,r,e)$ and $q=\sum_{i\leq
e}p_e$, and effect all the codings required by the suffix table. That is, for
every entry $(x',y')$ associated with $(x',y',r',e')$ and $q'$, place $w'1^{\la
x',y',r',e'\ra}0^{q'(|w'|)}$ in $A_{y'}$ for every $w'\in A_{x'}$.

Next, verify:
\begin{itemize}
\item
We can simulate $M_e(Ay;w)$ in $2^m$ time.
\item
Placing $w$ into $A_x$ and effecting again all the codings required by the
suffix table will not place any string shorter than $p_e(m)$ into $A_y$.
\end{itemize}
If either fails, do nothing further. Otherwise let $A_x(w)=1-M_e(A_y;w)$,
declare $P_{xye}$ satisfied and effect all codings required by the suffix table.

This completes the construction. It remains to verify \eqref{enum:SI2finite},
\eqref{enum:SI2infinite} hold, and that $A_x$ is indeed EXPTIME.

Suppose $|V_{xy}|=i$ and that the $i$th element enters $V_{xy}$ at step $s$.
Let $n$ be the stage where $g(n)=\la x,y,s,i\ra$. At stage $n$ the suffix table
is updated to associate
$(x,y)$ with $(x,y,s,i)$ and $q=\sum_{j\leq i}p_j$, meaning for any $w$ already
in $A_x$, $w1^{\la x,y,s,i\ra}0^{q(|w|)}$ is in $A_y$.
Observe that as no element enters $V_{xy}$ again, the $(x,y)$ entry of the table
is never updated.
This means that if some other $w$ ever enters $A_x$, effecting
the suffix table will require that $w1^{\la x,y,s,i\ra}0^{q(|w|)}$ be in $A_y$.
This gives a polynomial time reduction from $A_x$ to $A_y$: 
\[
f(w)=
\begin{cases}
w1^{\la x,y,s,i\ra}0^{q(|w|)}&\text{if }|w|\geq g(n)\\
\text{least }v\in A_y&\text{if }|w|<g(n)\text{ and }w\in A_x\\
\text{least }u\notin A_y&\text{if }|w|<g(n)\text{ and }w\notin A_x.
\end{cases}
\]
Note that the two latter cases are in
fact constant time: the amount of time to find $v$ or $u$ will be the same on
any input, and since $|w|$ is bounded by $g(n)$ the time to determine whether
$w\in
A_x$ is also constant. This establishes
\eqref{enum:SI2finite}.

Next, suppose $V_{xy}$ is infinite. We will show that the function computed by $M_e$ is not a reduction
from
$A_x$ to $A_y$. As $e$ is arbitrary, this will show that no machine is a
reduction, and hence a reduction cannot exist. We will do this via the
requirements $P_{xye}$: first by showing that once a requirement is satisfied,
it is never again violated and second by showing that with $V_{xy}$ infinite
$P_{xye}$ is eventually satisfied for every $e$.

Suppose that $P_{xye}$ is declared satisfied at stage $n$.
Observe that $P_{xye}$ is not injured at stage $n$: we only effect the suffix
table codings if they do not place any string shorter than $p_e(|w|)$ into
$A_y$, and as $p_e$ is the clock of $M_e$, the computation on $M_e(w)$ cannot
possibly query the oracle for any strings longer than than $p_e(|w|)$. As such
whatever strings may have been added to $A_y$ through the codings do not affect
the computation of $M_e$. So if $A_x(w)=1-M_e(A_y;w)$ held when $P_{xye}$ was
declared satisfied, it still holds after the codings have been effected.
Neither can $P_{xye}$ be injured at any later stage: we have verified that
$M_e(A_y;w)$ can be simulated in $2^{|w|}={g(n+1)}$ time, so at the next
stage of the
construction whatever strings may be added to $A_y$ will be too large to affect
the computation of $M_e(A_y;w)$.

Next, suppose that $P_{xye}$ has never been declared satisfied. We will show
that there exists an $e'>e$ with $M_{e'}(w')=M_e(w')$ for all $w'$, such that
$P_{xye'}$ has been
declared satisfied. This will ensure that $P_{xye}$ is, in fact, eventually
satisfied, even though it is never declared to be so.

We wish to pad $M_e$ into a large enough $M_{e'}$ so that at some stage $\la
x,y,r,e'\ra$
satisfies all the verifications of the construction and therefore $A_x(w)\neq
M_{e'}(A_y;w)$ and $P_{xye'}$ is satisfied. By padding $M_e$ we mean we add
inaccessible states so that the machine index changes but the machine behaviour
remains the same on any string.

This means we need to show that for a large enough $e'$, at stage $m=\la
x,y,r,e'\ra$, the following hold:
\begin{itemize}
\item
We can simulate $M_{e'}(Ay;w)$ in $2^m$ time.
\item
Placing $w$ into $A_x$ and effecting all the codings required by the
suffix table will not place any string shorter than $p_{e'}(m)=p_e(m)$ into
$A_y$.
\end{itemize}

The first will hold eventually as $p_{e'}$ grows polynomially with $m$
while $2^m$ grows exponentially.
The second is violated only if there is a sequence of entries in the suffix
table $(z_0,z_1),(z_1,z_2),\dots,(z_k,z_{k+1})$, $x=z_0$ and $y=z_{k+1}$,
with corresponding polynomials $q_0,\dots,q_k$ satisfying
$q_0(m)+q_1(m_1)\dots+q_k(m_k)\leq p_{e'}(m)$
where $m<\dots<m_k$. That is, placing $0^m$ in $A_x=A_{z_0}$ will place
$0^m1^{\la z_0,z_1,s,i\ra}0^{q_0(m)}$ in $A_{z_1}$, which in turn will add
another suffix and put the resulting string in $A_{z_2}$, and so on through the
entire chain, with the final string being shorter than $p_e(m)$. Clearly a
necessary (but not sufficient) condition for this is that all the polynomials
$q_i$ be smaller than $p_e$. We will show that this will eventually fail to
hold. Note that there is no point in considering the codings effected at later
stages in the construction, as the strings considered there will be much too
large.

Observe that by the transitivity of $\preceq$, $V_{z_i,z_{i+1}}$ must be
infinite for some $i$. Suppose $d>e$ and the $d$th element enters
$V_{z_i,z_{i+1}}$ at step $s$.
At the stage of the construction corresponding to $\la z_i,z_{i+1},s,d\ra$ we
will update the $(z_i,z_{i+1})$ entry of the suffix table with $q_i>p_e$. As
$p_e=p_{e'}$, $q>p_{e'}$.
This will ensure that $q_0(m)+q_1(m_1)\dots+q_k(m_k)> p_{e'}(m)$. As such if
$e'>d>e$, it will be large enough to satisfy $P_{xye'}$, and therefore
$P_{xye}$. This establishes \eqref{enum:SI2infinite}.

Finally, we must show that $A_x$ is EXPTIME. Observe that any string in $A_x$
begins with $0^{g(n)}$ and is followed by zero or more suffixes of the form
$1^i0^j$. If a string is of the wrong form, we can reject it immediately.

First, perform the first $n$ steps of the construction to determine whether
$0^{g(n)}\in A_z$, and to construct the suffix table at that stage.
We will demonstrate that the $n$th step can be performed in exponential time,
and accordingly so can the first $n$.
With $g(n)=\la z,y,r,e\ra$ we begin by constructing $V_{xy,r}$. As $r\leq g(n)$,
this takes at most polynomial time.
Next we check whether $M_e(A_y;0^{g(n)})$ can be run in $2^m$ steps: this can be
done simply by calculating $p_e(g(n))$.
After that we verify whether placing $0^{g(n)}$ in $A_z$ will place any strings
shorter than $p_e(g(n))$ in $A_y$.
Observe that if we treat the suffix table as an adjacency table for a graph,
this corresponds to finding a path of the lowest weight.
As the graph has at most two nodes for every $g(k)$, $k<m$, it has very much
less nodes than $g(n)$ and therefore than the length of the input string.
As such finding such a path can be done in well under exponential time. Finally
we need to simulate $M_e(A_y;0^{g(n)})$.
Note that as we have already performed all the steps prior to $n$, we know which
strings $A_y$ contains so the simulation only takes polynomial time.

At this stage if we have determined that $0^{g(n)}\notin A_z$, we can safely
conclude that the queried string is not in $A_x$.
Otherwise we simply need to verify whether every suffix of the string
corresponds to a valid entry in the suffix table. Specifically, if the suffix is
of the form $1^{\la x',y',r',e'\ra}0^q$ we perform the first $n'$, $g(n')=\la
x',y',r',e'\ra$, steps of the construction and verify whether the suffix table
entry for $(x',y')$ does, in fact, associate $(x',y')$ with $( x',y',r',e')$
and $q$. By using the same argument as above we can, mutatis mutandis, verify
that every such verification takes no more than exponential time. As there are
at most linearly many verifications to make, the entire process is within
EXPTIME.
\end{proof}

\begin{corollary}
Polynomial time equivalence of EXPTIME sets is complete for $\SI 2$ equivalence
relations.
\end{corollary}
\begin{proof}
Follows from Proposition~\ref{prop:preorders}, as polynomial time equivalence is
the symmetric fragment of polynomial time reducibility.
\end{proof}

\section{\texorpdfstring{$\mathbf{\SI 3}$}{Sigma 03} and Isomorphism of
Groups}\label{sec:Sigma03}

There are two parts to this section. In \ref{sec:almost_equal} we show that
almost inclusion (equality) of r.e.\ sets is complete for $\SI 3$ preorders
(equivalence relations). In \ref{sec:subgroups} we
show that this result is in fact equivalent to embeddability (isomorphism) of
computable subgroups of $(\QQ,+)$, as a consequence of Baer's (\cite{Baer1937})
characterisation of the subgroup structure of $(\QQ,+)$.

\subsection{Almost equality of r.e.\ sets}\label{sec:almost_equal}

We will make use of a number of results from the literature of recursion theory.
To facilitate that end we will first introduce some notational conventions for
this section.

We will use $A\subset_m B$ and $A\subset_{sm} B$ to denote that $A$ is a major
and small-major subset of $B$ respectively. See \cite{Soare1987} chapter X:4 for
these
notions. Maass and Stob (\cite{Maass1983}) have shown that given any $A\subset_m
B$ the lattice generated by almost inclusion on the r.e.\ sets between them is
unique up to isomorphism. We will use $[A,B]$ to denote this lattice, and $C^*$
to denote the equivalence class of $C$ under $=^*$. If $A$ is non-recursive, and
$A\setminus B$ is r.e.,\ then $B$ is said to be a split of $A$, which we will
denote with $B\sqsubseteq A$.

We are now ready to prove the theorem.

\begin{theorem}\label{thm:almost_inclusion}
Almost inclusion of r.e.\ sets is complete for $\SI 3$ preorders. That is, for
any
$\SI 3$ preorder $P$ there exists a computable $f$ such that:
$$Pxy \iff W_{f(x)}\subseteq^*W_{f(y)}.$$
\end{theorem}
\begin{proof}
Fix a non-recursive $A$. As $A$ is non-recursive there exists a $D\subset_{sm}
A$ (see \cite{Soare1987}, page 194). Then for every $X\sqsubseteq A$ we have the
following (\cite{Nies1998}, Lemma 4.1.2):
$$X\subseteq^*D\iff X\text{ is recursive}.$$
We can then define the Boolean algebra:
$$\mathcal{B}_D(A)=\{(X\cup D)^*\ |\ X\sqsubseteq A\},$$
with $D^*$ and $A^*$ as the minimum and maximum elements respectively, the
complement of $(X\cup D)^*$ being $((A\setminus X)\cup D)^*$, and $(X\cup
D)^*\wedge (Y\cup D)^*=((X\cap Y)\cup D)^*$. Note that we are guaranteed the
existence of complements because $X$ ranges over splits.

The Friedberg splitting theorem (\cite{Friedberg1958}) implies that any
non-recursive $X$ can be split into non-recursive $X_1,X_2$ obtained uniformly
from an r.e.\ index for $X$. By iterating this process we can obtain a uniform
sequence of splits $X_n\sqsubseteq A$. If we use $p_n$ to denote $(X_n\cup
D)^*$, $\mathcal{F}$ will denote the subalgebra generated by the sequence
$\{p_n\ |\ n\in\NN\}$.

Let $P$ be an arbitrary $\SI 3$ preorder, and $\mathcal{I}_0$ the ideal
generated by $\{p_n\wedge\neg p_k\ |\ Pnk\}$. We claim that:
$$Pnk\iff p_n\wedge\neg p_k\in\mathcal{I}_0.$$
Left to right is clear from definition. For the other direction let
$\mathcal{B}_P$ be the boolean algebra generated by subsets of $\NN$ of the form
$\hat{i}=\{r\ |\ Pri\}$. Consider the Boolean algebra homomorphism
$g:\mathcal{F}\ria\mathcal{B}_P$ induced by the map $p_i\mapsto\hat{i}$. As
$g(p_n)\subseteq g(p_k)$ whenever $p_n\wedge\neg p_k\in\mathcal{I}_0$, it
follows that $g$ maps $\mathcal{I}_0$ to $\emptyset$. If $\neg Pnk$ then
$\hat{n}\nsubseteq\hat{k}$, $g(p_n\wedge\neg p_k)$ is non-empty and hence
$p_n\wedge\neg p_k\notin\mathcal{I}_0$.

Now let $\mathcal{I}$ be the ideal of $\mathcal{B}_D$ generated by
$\mathcal{I}_0$. The above argument then implies:
$$Pnk\iff ((X_n\setminus X_k)\cup D)^*\in\mathcal{I}.$$
Note that $\mathcal{I}$ is a $\SI 3$ ideal. By the base case of the ideal
definability lemma (\cite{Harrington1998}) there exists a $B\in[D,A]$
satisfying:
$$(X\cup D)^*\in\mathcal{I}\iff X\subseteq^* B,$$
which of course means:
$$Pnk\iff ((X_n\setminus X_k)\cup D)^*\in\mathcal{I}\iff X_n\subseteq^* X_k\cup
B.$$
This gives us the desired reduction.
\end{proof}

\subsection{Subgroups of
\texorpdfstring{$\mathbf{(\QQ,+)}$}{(Q,+)}}\label{sec:subgroups}

In this section our goal is to use the result above to show that the isomorphism
of
computable subgroups of
$(\QQ,+)$ is complete for $\SI 3$ equivalence relations. Before we proceed, we
first need to establish some facts about subgroups of $(\QQ,+)$. Namely,
such subgroups are isomorphic if and only if they agree on all but a finite
number of prime powers. Likewise, there exists an embedding from $A$ to $B$ if
and only if the set of prime powers in $A$ is almost included in the set of
prime powers in $B$. These results
are based on \cite{Baer1937}. We give simplified proofs to better tailor them to
the computational setting. For notational convenience, we treat 1 as the 0th
prime.

\begin{lemma}\label{lemma:subgroup_generators}
Every subgroup of $(\QQ,+)$ is isomorphic to one with a generating set composed
of prime powers; specifically, the negative powers. That is, every subgroup
is isomorphic to one generated by
elements of the form
$1/p_i^k$ where $p_i$ is the $i$th prime and $k$ is some integer.
\end{lemma}
\begin{proof}

First observe that every subgroup of $(\QQ,+)$ is isomorphic to a subgroup
that contains 1: if $A$ be a subgroup and $s/t\in A$, then it is easy to verify
that
$f(x)=t/s\cdot x$ is an isomorphism of subgroups, and as $f(s/t)=1$ we have the
desired
subgroup in $f(A)$. A subgroup of $(\QQ,+)$ is then either isomorphic to
$(\ZZ,+)$, in which case it is clearly generated by $\{1/p_0\}$, or it contains
elements strictly between 0 and 1.

Such a group's generating set, without loss of generality, consists of $g$ for
$0<g\le 1$: any element $x>1$ can be obtained from $x-n$ by adding 1 $n$
times, and for the right choice of $n$, $x-n$ will lie between 0 and 1.

Next we claim that for coprime $a$ and $b$, $a/b\in A$ if and only if $1/b\in
A$, so in fact we can assume the generating set consists of elements with 1 in
the numerator. This
is because if
$a$ and $b$ are coprime then there exist integers $x,y$ such that $ax+by=1$. As
$a/b\in A$, so is $xa/b$ and as $1\in A$ so is $y$. We can then obtain $1/b$ as
follows:
$$\frac{xa}{b}+y=\frac{ax+by}{b}=\frac{1}{b}.$$
For the other direction, $1/b\in A$ immediately implies $a\cdot 1/b=a/b\in A$.

We now know that any subgroup is isomorphic to one generated by $1$ along with
fractions of the form $1/q$. The last step is to show that such a subgroup is
also generated by
fractions of the form $1/p^k$, where $p^k$ appears in the prime
factorisation of
some such $q$.

Let $1/q$ be an arbitrary element of the group, and let $q=p^kr$ for some prime
$p$, where $r$ and $p$ are coprime. It immediately follows that $1/p^k$ is in
the group, as $1/p^k=r/q$. For the
other direction, suppose we have a group containing $1/p^k$ and $1/r$. As $p^k$
and $r$ are coprime, there exist integers $x,y$ such that $xp^k+yr=1$. Observe
that:
$$\frac{x}{r}+\frac{y}{p^k}=\frac{xp^k+yr}{p^kr}=\frac{1}{q}.$$
By induction, it follows that the only elements necessary in the generating set
are the reciprocals of the prime powers appearing in the decomposition of $q$.
\end{proof}

In light of this result, from now on when we speak of ``subgroups of $(\QQ,+)$"
we will mean those subgroups generated by a what we will term a \emph{standard}
generating set. Denoted by $S(A)$, the standard generating set of $A$ contains
$1/p^k$ for all primes $p$ and integers $k$ such that $1/p^k\in A$. If $A$ and
$B$ are two subgroups, we use $S(A)-S(B)$ to denote the set consisting of
$1/p^{k-j}$ where $k,j$ are the largest integers, if such exist, such that
$1/p^k$ and $1/p^j$ are in $S(A)$ and $S(B)$ respectively. Note that $S(A)-S(B)$
is finite if $S(A)\subset^*S(B)$.

\begin{proposition}\label{prop:subgroup_isomorphism}
For subgroups $A,B$ of $(\QQ,+)$, $A\cong B$ if and only if $S(A)=^*S(B)$.
\end{proposition}
\begin{proof}
First, suppose $S(A)=^*S(B)$, let $S(A)-
S(B)=\{1/p_1^{k_1},\dots,1/p_n^{k_n}\}$ and $S(B)-
S(A)=\{1/q_1^{j_1},\dots,1/q_n^{j_n}\}$. Note that there is no generality lost
in assuming both sets are of cardinality $n$ as we can always set the
appropriate exponents to 0. We claim that the desired isomorphism is $f:A\ria B$
given by:
\[f(x)=\frac{p_1^{k_1}\dots p_n^{k_n}}{q_1^{j_1}\dots q_n^{j_n}}x.
\]

It is easy to see that $f$ preserves identity and the group operation simply as
a consequence of arithmetic on the rational numbers:
$$f(0)=\frac{p_1^{k_1}\dots p_n^{k_n}}{q_1^{j_1}\dots q_n^{j_n}}0=0,$$
$$f(x+y)=\frac{p_1^{k_1}\dots p_n^{k_n}}{q_1^{j_1}\dots
q_n^{j_n}}(x+y)=\frac{p_1^{k_1}\dots p_n^{k_n}}{q_1^{j_1}\dots
q_n^{j_n}}x+\frac{p_1^{k_1}\dots p_n^{k_n}}{q_1^{j_1}\dots
q_n^{j_n}}y=f(x)+f(y).$$
It remains to
show that $f$ is bijective, and that the image of $f$ is indeed contained in
$B$.

To see that $f(x)\in B$ for all $x$, observe that $s/t\in B$ if and only if
all prime powers $q^m$ that divide $t$, $1/q^{m}\in S(B)$. It follows that for
any $u/v\in A\setminus B$, there must exist a $p^i\in S(A)\setminus S(B)$ $v$
such that $p^i\ |\ v$. Without loss of generality, let this be the largest $i$
with this property. Since $S(A)=^*S(B)$, there must exist a largest $k$ such
that $1/p^k\in S(A)$ and a largest $j$ such that $1/p^j\in S(B)$. It follows
from our definition of $f$ that $f$ multiplies $u/v$ by $p^{k-j}$, and as such
the exponent of $p$ in the denominator of $f(u/v)$ will be $p^{i-k+j}$, and
$1/p^{i-k+j}\in S(B)$. This allows us to conclude that $f(x)\in B$.

To see that $f$ is onto, let $s/t$ be an arbitrary element of $B$. We need to
show that there exists a $u/v\in A$ satisfying:
$$\frac{p_1^{k_1}\dots p_n^{k_n}u}{q_1^{j_1}\dots
q_n^{j_n}v}=\frac{s}{t}.$$
Let $t=q_1^{z_1}\dots q_n^{z_n}P$ where $z_i\le j_i$ and $P$ is such that
$1/P\in A$. Observe that:
$$\frac{p_1^{k_1}\dots p_n^{k_n}}{q_1^{j_1}\dots
q_n^{j_n}}\cdot\frac{q_1^{j_1-z_1}\dots q_n^{j_n-z_n}s}{p_1^{k_1}\dots
p_n^{k_n}P}=\frac{s}{t}.$$
The right multiplicand is clearly in $A$ as the denominator only contains primes
from $S(A)$.

Finally, the fact that $f$ is one to one follows immediately from rational
arithmetic.

For the second direction suppose that
$S(A)\setminus S(B)$ is infinite. To proceed, we observe that we can define a
notion of divisibility in a group in the standard fashion:
$$x\ |\ y\iff zx=y\text{ for some }z$$
Clearly, any isomorphism from $A$ to $B$ must preserve divisibility. I.e., $x\
|\ y$ if and only if $f(x)\ |\ f(y)$, and moreover if $zx=y$ then $zf(x)=f(y)$.

Let $f$ then be some function from $A$ to $B$. If $f$ is to be an isomorphism,
for every $1/p^k\in S(A)\setminus S(B)$ we must have $f(1)=p^kf(1/p^k)$. That
is, if $f(1)=s/t$ and $f(1/p^k)=u/v$ we have:
$$\frac{s}{t}=\frac{p^ku}{v}.$$
Now, observe that $u/v\in B$ and $1/p^k\in S(A)\setminus S(B)$ implies that
if $p^j$ appears in the prime decomposition of $v$, then $j<k$ and hence
$p^{k-j}$ must appear in the decomposition of $s$. However, as there are
infinitely many elements in $S(A)\setminus S(B)$ this would mean that $s$ would
need to be infinitely large, which is impossible.
\end{proof}

A simpler version of the same reasoning can be adapted to show that subgroup
embeddability is almost
inclusion of standard generating sets.

\begin{proposition}\label{prop:subgroup_embeddability}
For subgroups $A,B$ of $(\QQ,+)$, $A$ is embeddable in $B$ if and only if
$S(A)\subseteq^*S(B)$.
\end{proposition}
\begin{proof}
Suppose $S(A)\subseteq^*S(B)$ and let
$S(A)-S(B)=\{1/p^{k_1}_1,\dots,1/p^{k_n}_n\}$. We claim that
$f(x)=p^{k_1}_1\dots p^{k_n}_nx$ is an embedding of $A$ into $B$. This is
clearly one to one and preserves both addition and the group operation. The
image of $f$ lies within $B$ as multiplying by the primes in $S(A)-S(B)$ will
rid the denominator of any primes inadmissible in $B$.

For the other direction suppose $S(A)\setminus S(B)$ is infinite. Again the
preservation of divisibility induces that if $zx=y$ then $zf(x)=f(y)$. Since
there are infinitely many members of $S(A)\setminus S(B)$, we cannot obtain
this.
\end{proof}

\begin{corollary}\label{cor:subgroup_iso_bi}
Subgroups of $(\QQ,+)$ are isomorphic if and only if they are bi-embeddable.
\end{corollary}
\begin{proof}
One direction is clear: an isomorphism is an embedding both ways, so if $f$ is
an isomorphism between $A$ and $B$, $f$ and $f^{-1}$ are the required
embeddings.

Next, suppose $A,B\subseteq(\QQ,+)$ are bi-embeddable via $f:A\ria B$ and
$g:B\ria A$. We claim that this implies that $S(A)\subseteq^*S(B)$ and
$S(B)\subseteq^*S(A)$, hence $S(A)=^*S(B)$ which, by
Proposition~\ref{prop:subgroup_isomorphism}, implies an isomorphism.
\end{proof}

We now have everything we need to speak about computable groups.

\begin{definition}\label{def:subgroups}
A \emph{computable subgroup of $(\QQ,+)$} is a 4-tuple of computable functions,
$( e,\oplus,\ominus,I)$. The natural numbers are taken to encode group elements,
not necessarily uniquely.  The function $e$ selects the identity element(s), meanings $e(x)=1$ for identity and 0
otherwise. The group operation is encoded by $\oplus:\NN\times\NN\ria\NN$ and
$\ominus:\NN\ria\NN$ takes an element to its inverse. $I$ is the interpretation
function, which maps group elements to (a fixed encoding of) the rationals. We
require that the group axioms be respected through $I$, namely:
\begin{itemize}
\item $I(\oplus(\oplus(a,b),c))=I(\oplus(a,\oplus(b,c)))$.
\item $(e(x)=1)\ria (I(\oplus(a,x))=I(\oplus(x,a))=I(a))$.
\item $(\oplus(a,\ominus(a))=x)\ria (e(x)=1)$.
\end{itemize}
As we are dealing with additive subgroups, we also require that $I$ respects
addition on the rational numbers:
\begin{itemize}
\item $I(\oplus(a,b))=I(a)+I(b)$.
\end{itemize}
\end{definition}

Of course, there is no way to enumerate arbitrary 4-tuples of computable
functions. However we will now show that we can identify with every r.e.\ set a
computable group, and all computable subgroups of $(\QQ,+)$ can be obtained in
this fashion. This allows us to speak of the $i$th computable subgroup, denoted
$G_i$, where $S(G_i)$ is encoded by $W_i$. As such if $( e,\oplus,\ominus,I)$ is
$G_i$, we will let $\la e,\oplus,\ominus,I\ra=i$.

\begin{lemma}\label{lemma:subgroup_representation}
A subgroup of $(\QQ,+)$ is computable if and only if we can uniformly
effectively obtain its standard generating set from some $W_i$.
\end{lemma}
\begin{proof}
We will treat $W_i$ as consisting of pairs $(x,y)$, closed downward in the
second component. We can then define $G_i$ to be the group where $S(G_i)$ is the
set containing $1/p_x^y$ for all $(x,y)\in W_i$. We can now show that $G_i$ is
computable.

Fix some encoding of finite sequences of integers. Let $w_s$ be the element that
enters $W_i$ at stage $s$ if such an element exists, or 0 otherwise. If $x$ is
the encoding of $z_0z_1\dots z_n$, let $I(x)=z_0w_0+z_1w_1+\dots+z_nw_n$. $I$ is
clearly computable, so there exists a partial recursive $I^{-1}$ where
$I^{-1}(y)$ is the least $x$ such that $I(x)=y$. This allows us to define
$\oplus(x,y)$ as $I^{-1}(I(x)+I(y))$, $\ominus(x)$ as $I^{-1}(-I(x))$ and let
$e(x)=1$ if and only if $I(x)=0$, and 0 otherwise. All of these are clearly
computable and satisfy the axioms of Definition~\ref{def:subgroups}.

For the other direction, suppose we are given $( e,\oplus,\ominus,I)$. We can
enumerate $W_i$ by running $I$ on every number and listing $(x,y)$ each time we
come across an element of the form $1/p_x^y$.
\end{proof}

This gives us all we need to construct a $\SI 3$-complete preorder.

\begin{theorem}\label{thm:Si03_subgroups}
Computable embeddability of computable subgroups of $(\QQ,+)$ is complete for
$\SI 3$ preorders.
\end{theorem}
\begin{proof}
First, let us verify that the relation is indeed $\SI 3$. We can represent it
as:

$$\{(\la e_1,\oplus_1,\ominus_1,I_1\ra,\la e_2,\oplus_2,\ominus_2,I_2\ra)\ |\
\exists i\
\varphi_i\text{ is total}\ \wedge$$
$$\forall xy\ \varphi_i(x)=\varphi_i(y)\ria x=y\ \wedge$$
$$e_1(x)=e_2(\varphi_i(x))\ \wedge$$
$$\oplus_1(x,y)=\oplus_2(\varphi_i(x),\varphi_i(y))\}.$$
Totality of p.r.\ functions is a $\PI 2$ property, so placing an existential
quantifier over it is $\SI 3$. By Lemma~\ref{lemma:subgroup_representation}, we
have an effective encoding of groups so we can obtain the code for $\oplus$,
$\ominus$ and $e$ effectively. As $\varphi_i$ is computable, the rest of the
formula is $\PI 1$ so with an existential quantifier over it is $\SI 2$, and
thus contained within $\SI 3$.

To show embeddability is complete we invoke Theorem~\ref{thm:almost_inclusion}.
Almost inclusion is complete for $\SI 3$ preorders, so
Proposition~\ref{prop:subgroup_embeddability} and
Lemma~\ref{lemma:subgroup_representation} already did all the work for us, and
we obtain:
$$W_x\subseteq^*W_y\iff G_x\text{ is embeddable in }G_y.$$\end{proof}

It immediately follows that computable bi-embeddability of subgroups of
$(\QQ,+)$ is complete for $\SI 3$ equivalence relations. However, we can
strengthen this result somewhat. By Corollary~\ref{cor:subgroup_iso_bi}
bi-embeddability implies isomorphism, and there is no need to speak of a
computable isomorphism because that follows immediately.

\begin{proposition}\label{prop:subgroup_comp_iso}
Computable subgroups of $(\QQ,+)$ are isomorphic if and only if they are
computably isomorphic.
\end{proposition}
\begin{proof}
One direction is clear. For the other direction observe that the isomorphism
constructed in Proposition~\ref{prop:subgroup_isomorphism} involves multiplying
by a rational number. This is clearly computable in our framework: if $(
e_1,\oplus_1,\ominus_1,I_1)$ and $( e_2,\oplus_2,\ominus_2,I_2)$ are isomorphic
via $f(x)=qx$, we can compute this by $I_2^{-1}(qI_1(x))$.
\end{proof}
\begin{corollary}
Isomorphism of computable subgroups of $(\QQ,+)$ is complete for $\SI
3$ equivalence relations.
\end{corollary}
\begin{proof}
Theorem~\ref{thm:Si03_subgroups} establishes that computable bi-embeddability of
computable subgroups is complete for $\SI 3$ equivalence relations.
Propositions~\ref{cor:subgroup_iso_bi} and \ref{prop:subgroup_comp_iso} allow
us to lift this result to isomorphism of computable subgroups.
\end{proof}

On the other hand, the restriction to computable subgroups is indeed necessary:
there are more isotypes of subgroups of $(\QQ,+)$ than there are natural
numbers, so some subgroups are inherently uncomputable.

\begin{proposition}
There are uncountably many isotypes of subgroups of $(\QQ,+)$. As such, there
exist uncomputable subgroups of $(\QQ,+)$, even modulo isomorphism.
\end{proposition}
\begin{proof}
We will prove the proposition in two steps. First, we will show that there
exists an injection from
$2^\NN/=^*$ to $Sub(\QQ,+)/\cong$, where $Sub(\QQ,+)$ is the set of subgroups of
$(\QQ,+)$. That is, there are at least as many isotypes of subgroups of
$(\QQ,+)$ as
there are equivalence classes of $=^*$ on the natural numbers. Next, we will
show that there is an injection from $2^\NN$ into $2^\NN/=^*$. I.e., the
cardinality of the powerset of the naturals is no larger than the cardinality of
the powerset of the naturals modulo almost equality. This will establish that
there are uncountably many isotypes of subgroups of $(\QQ,+)$, and therefore
they cannot all be computable.

First, let $f:2^\NN\ria Sub(\QQ,+)$ be the function which maps the set $A$ to
the
subgroup generated by
$\{1/p_i\ |\ i\in A\}$. It is easy to see that $f$ is one to one: suppose $n\in
A$, $n\notin B$. We claim that $1/p_n\in f(A)$, $1/p_n\notin f(B)$ and hence
$A\neq B$ implies $f(A)\neq f(B)$. Suppose to the contrary, there are integers
$z_0,\dots,z_k$ and indices $i_0,\dots,i_k$ with $i_j\neq n$ for all $j$ such
that:
$$\sum_{j=0}^k\ \frac{z_j}{p_{i_j}}=\frac{1}{p_n},$$
$$\frac{z_0(p_{i_1}\dots p_{i_k})+\dots+z_k(p_{i_0}\dots
p_{i_{k-1}})}{p_{i_0}\dots p_{i_k}}=\frac{1}{p_n},$$
$$p_n(z_0(p_{i_1}\dots p_{i_k})+\dots+z_k(p_{i_0}\dots
p_{i_{k-1}}))=p_{i_0}\dots p_{i_k},$$
but that is clearly impossible.

Next, let $g:2^\NN/=^*\ria Sub(\QQ,+)/\cong$ be the function which maps
$[A]_{=^*}$ to $[f(A)]_{\cong}$. The fact that this is an injection follows
immediately from the observation above and
Proposition~\ref{prop:subgroup_isomorphism}.

Finally, let $h:2^\NN\ria 2^\NN/=^*$ be the function which maps $A$ to $\{x\ |\
x=e_n\text{ for some }n\in A\text{, or }x=o_{p_n^k}\text{ for some }k\in\NN,n\in
A\}$, where $e_n$ is the $n$th even number (i.e. $2n$) and $o_n$ is the $n$th
odd number ($2n+1$). For intuition, this could
be thought of as $h(A)=\la A,A'\ra$ where $A'$ is the set that contains all
powers of the $n$th prime for every $n$ in $A$, with the first component of
$h(A)$ encoded by the even numbers and the second by the odd.

To see that $h$ is one to one suppose $n\in A$, $n\notin B$. We want to show
that $h(A)\triangle h(B)$ is infinite, but this is immediate as $o_{p_n^k}\in
h(A)$ for all $k$, while $o_{p_n^k}\notin h(B)$.
\end{proof}

\chapter{The case of \texorpdfstring{$\mathbf{\PI
n}$}{Pi 0n}}\label{sec:Pi0nrelations}

There are two results in this section. In \ref{sec:Pi01} below we demonstrate
the existence of a $\PI 1$-complete
equivalence relation and preorder, and show how a natural example can be
obtained from polynomial time trees. The existence of a $\PI 1$-complete
equivalence relation is based on a characterisation of every $\PI 1$ equivalence
relations as the kernel of a computable function. This result can actually be
derived from a recent work of Cholak, Dzhafarov, Schweber and Shore on partial
orders (\cite{Cholak2011}). We present both results below as the works are
independent and based on different arguments.

In \ref{sec:Pi0n} we present a quite different, and striking result: the
non-existence of equivalence relations complete
for $\PI n$, where $n\geq 2$. As it happens, the same construction will also
suffice
to establish the non-existence of $\DE n$-complete equivalence relations as
well.

\section{\texorpdfstring{$\mathbf{\PI 1}$}{Pi 01}-complete equivalence
relations and preorders}\label{sec:Pi01}

While the existence of $\SI n$-complete equivalence relations was never in
question, with $\PI 1$ the case is different. In the absence of an effective
means of enumerating the class the technique used to obtain a $\SI n$-complete
equivalence relation did not work, and the question of the existence of any
complete equivalence
relation, no matter how artificial, was an open and interesting one. In light of
this fact we will present the results in a different order to that in
Section~\ref{sec:Si0nrelations}. Rather than letting the existence of a complete
equivalence relation follow from the existence of a complete preorder, in
\ref{sec:Pi01eqrel} below we will construct a complete equivalence relation
directly, demonstrating how a $\PI 1$ equivalence relation can be seen as a
limit of a family of computable approximations. Then in \ref{sec:Pi01preorder}
we will show how the characterisation of $\PI 1$ preorders of \cite{Cholak2011}
can be used to construct a $\PI 1$-complete preorder by the same principles.

\subsection{A complete equivalence relation}\label{sec:Pi01eqrel}

We begin with the observation that if we fix a computable function $f$, a
natural example of a $\PI 1$ equivalence
relation could look something like this:
\begin{equation}\label{eq:Pi01_eqrel_form}
\{(x,y)\ |\ \forall n\ f(x,n)=f(y,n)\}.
\end{equation}
In fact, the reader may find it difficult to think of a $\PI 1$ equivalence
relation that cannot be interpreted in this way. As it turns out, that is
because no such equivalence relations exist: every $\PI 1$ equivalence relations
arises from the column equality of some computable function. This result, which
we will shortly present as Theorem~\ref{thm:Pi01complete}, proves to be
the key step in constructing a $\PI 1$-complete equivalence relation.

We first introduce a convenient way of thinking about a $\PI 1$ equivalence
relation. As every such relation is the complement of an r.e.\ set, for any
$E\in\PI 1$ there exists some $W_e$ such that $E=\NN^2\setminus W_e$.
This allows us to construct a computable sequence of approximations to $E$,
$\{E_t\ |\ t\in\NN\}$, satisfying:
\begin{equation}\label{eq:Etinclusion}
E_{t+1}\subseteq E_t.
\end{equation}
\begin{equation}
E_t\cap [0,t]^2\text{ is an equivalence relation over }[0,t]^2.
\end{equation}
\begin{equation}\label{eq:Eintersect}
E=\bigcap_{t\in\mathbb{N}} E_t.
\end{equation}
This is done simply by enumerating $W_e$ until $(\NN^2\setminus W_e)\cap
[0,t]^2$ is an equivalence relation over $[0,t]^2$, which it must do eventually.
We will use $E'_t$ to denote $E_t\cap [0,t]^2$.

With this in mind, we can show that every $\PI 1$ equivalence relation is of the
form in \eqref{eq:Pi01_eqrel_form}.

\begin{theorem}\label{thm:Pi01complete} Given a $\PI 1$
equivalence relation $E$, we can effectively obtain a
computable $f$
such that for all $x,y$:
\begin{equation}
Exy \iff \forall n\ f(x,n)=f(y,n).
\end{equation} 
\end{theorem}

\begin{proof}Intuitively, $f(x,n)$ is an approximation of
$\min[x]_E$. It will recurse on $f(z,n)$ where $z$ is the smallest element
satisfying $E'_{\max(x,n)}xz$. If the smallest such $z$ is $x$ itself, then
$f(x,n)=x$. Observe that if $x<n$ then $f(x,n)=\min[x]_{E'_n}$. The idea is that
for a large enough $n$, $f(x,n)$ will in fact be $\min[x]_E$, so $Exy$ thus
implies $f(x,n)=f(y,n)$. If $n$ is not large enough, we will show that $Exy$
nevertheless implies $f(x,n)=f(y,n)$, namely $\min[z]_{E'_n}$ for some $z$.

Explicitly $f$ is:
\[
 f(x,n) =
  \begin{cases}
   f(z,n)& \text{for }z=\min[x]_{E'_{\max(x,n)}},\text{ if }z<x\\
   x& \text{otherwise. }
  \end{cases}
\]

First let us verify that $\neg Exy$ implies that for some $n$, $f(x,n)\neq
f(y,n)$. Observe that if $(x,y)\notin E$, then for a large enough $n$,
$(x,y)\notin
E'_n$. We can without loss of generality assume that this $n$ is larger than $x$
or $y$, so since $f(x,n)=\min[x]_{E'_n}$ and $f(y,n)=\min[y]_{E'_n}$,we have
$f(x,n)\neq f(y,n)$ as required.

It
remains to
consider the case where $(x,y)\in E$. In this case $(x,y)\in
E_n$ for all $n$ by \eqref{eq:Eintersect}.

We proceed by double induction on the recursive stack depth of $f(x,n)$ and
$f(y,n)$, $i$ and $k$, with the inductive hypothesis that for $x\le y$,
$(x,y)\in E_{\max(y,n)}$ implies $f(x,n)=f(y,n)$ for functions of depth at most
$i$ and $k$ respectively.

Base case: $i=k=0$. In this case $x=\min[x]_{E'_{\max(x,n)}}$ and
$y=\min[y]_{E'_{\max(y,n)}}$. Since $(x,y)\in E$, $(x,y)\in
E'_{\max(y,n)}$. It must be the case that $x=y$, so clearly $f(x,n)=f(y,n)$.

Inductive case for $i$: suppose the inductive hypothesis holds for functions of
depth $i$ and $k$. Let $f(x,n)$ be of depth $i+1$, $f(y,n)$ of at most $k$.
Let $f(x,n)=f(z,n)$. Observe that $f(y,n)=f(w,n)$, $z\le w\le x$: the upper
bound we get from $(x,y)\in E'_{\max(y,n)}$ , the lower from
\eqref{eq:Etinclusion}. We then have $(z,w)\in E'_{\max(w,n)}$ as
$E_{\max(x,n)}\subseteq E_{\max(w,n)}$, so we can apply the inductive
hypothesis on $f(z,n)$ and $f(w,n)$.

Inductive case for $k$: suppose the inductive hypothesis holds for functions of
depth $i$ and $k$. Let $f(y,n)$ be of depth $k+1$, $f(x,n)$ of at most $i$. Note
that $f(x,n)=f(z,n)$ and $f(y,n)=f(w,n)$ with $z\leq w$. It remains to show that
$(z,w)\in E_{\max(w,n)}$ before we can apply the inductive hypothesis. This is
clearly the case as $(y,w)\in E_{\max(y,n)}\subseteq E_{\max(x,n)}$. As
$E'_{\max(x,n)}$ is an equivalence relation and contains $(x,y),(y,w)$ and
$(x,z)$ it must also contain $(z,w)$. As $E_{\max(x,n)}\subseteq E_{\max(w,n)}$,
we obtain what is needed.\end{proof}

It should be noted that not only do we obtain $f(x,n)=\min[x]_E$ for a large
enough $n$,
but we also know that $n$ is large enough if $f(x,n)=x$. Recall that in
Proposition~\ref{prop:relations_minelements} we
have established that the set of least elements of a $\PI 1$ equivalence
relation is $\SI 1$ and hence r.e.: $f$ gives us the function which
enumerates them. Simply apply $f$ to all values of $x$ and $n$, and print $x$
whenever $f(x,n)=x$.

We can strengthen this result to hold for polynomial time functions, thereby
obtaining a $\PI 1$ complete equivalence relation as polynomial time functions
have an effective enumeration. We thank
Moritz M\"{u}ller at KGRC Vienna for suggesting the simplified proof.

\begin{theorem}\label{thm:quadratic_complete}
The equality of polynomial (in fact, quadratic) functions is complete for $\PI
1$ equivalence relations. That is, for every $\PI 1$ equivalence relation $E$
and for every $x,y$, we can effectively obtain quadratic time $g_{x},g_{y}$
such that:
\begin{equation}
Exy\iff g_{x}=g_{y}.
\end{equation}
\end{theorem}
\begin{proof}
The idea behind the proof is to pad the input until it is large enough to allow
us to put a bound on the running time of the $f$ we constructed in
Theorem~\ref{thm:Pi01complete}. We will thus construct a quadratic time $g$
mapping
from strings to strings such that
for a computable $p$ mapping integers to strings,
\begin{equation}
g(p(x),p(n))=f(x,n).
\end{equation}
This of course gives us:
\begin{equation}
Exy\iff\forall n\ g(p(x),p(n))=g(p(y),p(n)).
\end{equation}
We will then let $g_{x}(p(n))=g(p(x),p(n))$, and 0 if the input is not of the
form $p(n)$. Provided we can effectively verify if a string is of the form
$p(n)$ for some $n$, this gives us the theorem.

Let $E$ be a $\PI 1$ equivalence relation and $f$ be a computable function
mapping strings to strings over a binary alphabet, on input
$(1^x,1^n)$ behaving as on $(x,n)$ in   \ref{thm:Pi01complete}, and on any other
input outputting 0.
That is:
\[
 f(1^x,1^n) =
  \begin{cases}
   f(1^z,1^n)& \text{for }z=\min[x]_{E'_{\max(x,n)}},\text{ if }z<x\\
   1^x& \text{otherwise. }
  \end{cases}
\]
On any pair of strings not consisting entirely of ones $f$ outputs 0.

Without loss of generality, let $f(1^x,1^y)$ be computable in time $t(x,y)$
where  $t$ be an increasing function in $x$, $y$. This can be achieved, for
instance, by having $f(1^x,1^y)$ first compute the function value for all
smaller values of $x,y$ before beginning on $1^x$,$1^y$.

We want a time constructible $h$ satisfying $t(x,n)\le h(x)+h(n)$. This can be
achieved by having $h(k)$ run $f(1^k,1^k)$ and counting the number of steps.
This satisfies $t(x,n)\le h(x)+h(n)$ as for $x\le n$, $t(x,n)\le t(n,n)=h(n)$,
and $h(k)=m$ can be calculated in time quadratic in $m$: place a binary counter
at the start of the tape and simulate $f$ to the right of it.
In the worst case scenario for every step of $f$ the machine would iterate over
$m+\log m$ cells to update the counter, and whenever the counter gains another
bit the entire tape would need to be shifted right. This will happen $\log m$
times, and at most $m$ cells would need to be shifted each time. This would take
at most $m\log m$,
coming to $m\log m+m(m+\log m)$ or $O(m^2)$.

To compute $g(a,b)$ verify whether $a=1^x01^{h(x)}$ and $b=1^n01^{h(n)}$. If so
output $f(1^x,1^n)$, else output 0.
Observe that $g$ is quadratic time: we can verify whether $a=1^x01^z$ is of the
form
$1^x01^{h(x)}$ by beginning to compute $f(1^x,1^x)$, stopping whenever the
number of steps exceeds $z$.
If the input is of the right form we can compute $f(1^x,1^n)$ in time $t(x,n)\le
h(x)+h(n)$, which is of the same order as $|a|,|b|$ respectively.

Recall that we have defined $g_{x}(p(n))$ to be $g(p(x),p(n))$ for some
computable
$p$. Let $p(x)=1^x01^{h(x)}$. We have already seen that we can verify if a
string is of this form in time quadratic in the length of the string.
This ensures that for every $x,n$:
\begin{equation}
g(1^x01^{h(x)},1^n01^{h(n)})=f(1^x,1^n).
\end{equation} 
As for
all other values the function is 0, $g_{x}=g_{y}$ iff $\forall n\
f(1^x,1^n)=f(1^y,1^n)$, which by Theorem~\ref{thm:Pi01complete} guarantees that
$Exy$.
\end{proof}

As is often the case, having established the existence of one $\PI 1$ complete
relation, showing the existence of others is easy, as all we have to do is pick
a relation rich enough to encode computable function equality. For a natural
mathematical example, we will show that the isomorphism problem for polynomial
time
trees is $\PI 1$ complete.

\begin{definition}
A \emph{subtree of $\{0,\dots,c\}^*$} is a language closed under prefixes, $T$.
Implicit in the definition is the notion of a predecessor function, $pred:T\ria
T$, satisfying $pred(\epsilon)=\epsilon$ and for $w\neq\epsilon$, $pred(w)=v$
if and only if $va=w$ for some $a\in\{0,\dots,c\}$. An isomorphism of subtrees
is a bijective function preserving predecessor. A subtree $T$ is polynomial if
there is a procedure for verifying $w\in T$ in time polynomial in $|w|$.
\end{definition}

Note that we have an effective enumeration of all such trees. The set of
strings accepted by any polynomial time machine on the
alphabet $\{0,\ldots,c\}$ is,
by definition, a polynomial time language over $\{0,\ldots,c\}$. Every such
language induces a prefix-closed language:
given a machine $M$ and a string $s$, run $M$ on every prefix of $s$ and accept
iff all prefixes are accepted.
If $M$ is polynomial time this incurs at most a linear slowdown, and thus is
still polynomial.

Given an enumeration of polynomial time machines we will then say the tree
induced by $e$ to mean
the prefix-closed language induced by the $e$th machine. 

\begin{theorem}\label{thm:Pi01_trees}
Isomorphism of polynomial time subtrees of $\{0,\dots,c\}^*$ is
complete for $\PI 1$ equivalence relations.
\end{theorem}
\begin{proof}
First let us verify that the problem is indeed $\PI 1$.
We claim that the
isomorphism can be thus expressed by:
\begin{equation}\label{Pi01form}
\{ (e,i):\forall d\ T_{(e,d)}\cong T_{(i,d)}\}.
\end{equation}
By $T_{(e,d)}$ we mean the tree induced by $e$ on strings of length no greater
than $d$: that is, the subtree of the tree induced by $e$ of depth $d$.
The matrix is clearly computable, as it concerns isomorphism of finite
trees.
To show that the isomorphism of subtrees of any depth implies a complete
isomorphism we need K\"{o}nig's lemma.

First observe that as any isomorphism must preserve the predecessor relation, it
must
therefore map the root to the root.
We will define a tree $I$ in which there exists a branch of length $d$ starting
from the root for every isomorphism between $T_{(e,d)}$ and $T_{(i,d)}$,
and given such a branch the subbranch of length $d-1$ corresponds to the
restriction of that isomorphism to $T_{(e,d-1)}$ and $T_{(i,d-1)}$.
That is, the root of $I$ corresponds to the unique function mapping the sole
node of $T_{(e,0)}$ to the sole node of $T_{(i,0)}$,
and for every isomorphism between $T_{(e,d)}$ and $T_{(i,d)}$ there is a
corresponding node at depth $d$.
Since the structures are finite every node has finitely many children, as there
are only finitely many functions between them. But if \eqref{Pi01form} holds
there are infinitely many nodes in the tree.
By K\"{o}nig's lemma, there must be a branch of infinite length in $I$,
corresponding to a complete isomorphism between the trees induced by $e$ and
$i$.

Next we will show how given indices $e$ and $i$ of polynomial time machines
computing functions $f$ and $g$ respectively,
we can construct polynomial time subtrees of $\{0,1\}^*$, $T_e$ and $T_i$, such
that:\bc $T_e\cong T_i \LR \forall x\ f(x)=g(x).$\ec
This will be sufficient to establish $\PI 1$ completeness by
Theorem~\ref{thm:quadratic_complete}.

To construct $T_e$, whenever $M_e$ on input $x$ outputs $n$ we require that
$1^x0^n\in T_e$. We then close the resulting language under substrings.
This corresponds to an infinite 1-branch, with a 0-branch corresponding to the
machine output at each level, as illustrated in Figure \ref{Tree}.

\begin{figure}\label{Tree}
\begin{center}
\begin{tikzpicture}[dot/.style={draw,circle,minimum size=1mm,inner sep=0pt,outer
sep=0pt,fill=black}]
\node at (0,-0.5){$f(0)=3$};
\node at (0,-1){$f(1)=0$};
\node at (0,-1.5){$f(2)=1$};
\node at (0,-2){$\vdots$};

\node at (2.5,-1.5){$\RA$};

\path \foreach \x/\y/\k in { 5/0/0,  5.5/-0.5/1, 6/-1/2, 6.5/-1.5/3}
      {coordinate [dot] (p\k) at (\x,\y)};
\draw \foreach \k in {1,...,3}
      {(p0)--(p\k)};

\path \foreach \x/\y/\k in { 4.5/-0.5/0,  4/-1/1, 3.5/-1.5/2}
      {coordinate [dot] (q\k) at (\x,\y)};
\draw \foreach \k in {0,...,2}
      {(p0)--(q\k)};

\path {coordinate [dot] (q3) at (5.5,-1.5)};
\draw {(p2)--(q3)};

\node at (6.8,-1.7){$\ddots$};
\end{tikzpicture}
\end{center}
\caption{The construction of $T_e$}
\end{figure}
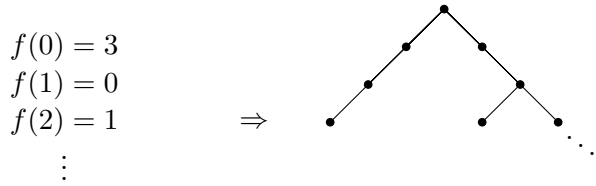

Let us verify that the resulting construction is indeed polynomial time.
Given a string we must first verify whether it is of the form $1^a0^b$ for some
$a,b\geq 0$, which can be done in linear time.
Then we check whether $M_e$ on input $a$ outputs $b$. Since $M_e$ is polynomial
time, so are these operations.

To see that $T_e\cong T_i \RA \forall x\ f(x)=g(x)$ we note again that the
isomorphism must map the root of $T_e$ to the root of $T_i$.
Having established this, assume for contradiction that $f(c)\neq g(c)$ for some
$c$, but $T_e\cong T_i$.
Observe that the isomorphism must map $1^c$ in $T_e$ to $1^c$ in $T_i$ as that
is the only node in both trees that is at distance $c$ from the root and has an
infinite branch below it.
However, this is impossible because one has a branch of length $f(c)$ and the
other of length $g(c)$.

To see that $T_e\cong T_i \LA \forall x\ f(x)=g(x)$ one need only observe that
with identical functions our construction will produce identical, and a fortiori
isomorphic, trees. \end{proof}

\subsection{A complete preorder}\label{sec:Pi01preorder}

Theorem~\ref{thm:Pi01complete} can be seen as a corollary of a largely unrelated
work. In \cite{Cholak2011} the authors in their investigation of partial orders
offer a characterisation of $\PI 1$ preorders (Proposition 3.1):
every such preorder is computably isomorphic to the
inclusion relation on a computable family of sets. As an equivalence relation is
merely a symmetric preorder, a $\PI 1$ equivalence relation is then isomorphic
to equality on a computable family of sets. This, of course, is another way to
view the equality
of computable functions.

As there is no effective enumeration of computable sets the inclusion relation
on them is not in itself a $\PI 1$ preorder. However, we can use a padding
argument as in \ref{thm:quadratic_complete} to extend the result of
\cite{Cholak2011} to polynomial time sets, which does result in a $\PI
1$-complete
preorder.

\begin{theorem}
The inclusion relation on polynomial time sets is complete for $\PI 1$
preorders. That is, taking $X_i$ to mean the $i$th polynomial time set, for
every $\PI 1$
preorder $P$ there exists a computable $f$
satisfying:
$$Pxy\iff X_{f(x)}\subseteq X_{f(y)}.$$
\end{theorem}
\begin{proof}
By Proposition 3.1 of \cite{Cholak2011}, for every $\PI 1$ preorder $P$ there
exists a sequence of computable sets $\{A_i\ |\ i\in\NN\}$ such that $Pxy$ if
and only if $A_x\subseteq A_y$. Our goal is to create a sequence of polynomial
time sets
of binary strings $\{X_i\ |\ i\in\NN\}$ and a computable $h$ satisfying:
$$n\in A_i\iff 1^n0^{h(n)}\in X_i.$$
Once we establish this the proof will be complete: by
\cite{Cholak2011} we shall have $Pxy$ if and only if $A_x\subseteq A_y$ which,
by our construction, holds if and only if $X_x\subseteq X_y$, and we can obtain
$X_i$ from $A_i$ computably via $h$.

As with equivalence relations, we will require computable approximations to our
preorder $P$. That is, a computable sequence $\{P_t\ |\ t\in\NN\}$, satisfying:
\begin{equation}\label{eq:Ptinclusion}
P_{t+1}\subseteq P_t.
\end{equation}
\begin{equation}
P_t\cap [0,t]^2\text{ is a preorder over }[0,t]^2.
\end{equation}
\begin{equation}\label{eq:Pintersect}
P=\bigcap_{t\in\mathbb{N}} P_t.
\end{equation}

The sequence $\{A_i\ |\ i\in\NN\}$ in \cite{Cholak2011} is constructed by
stages. We give their construction below:

At stage $t$, for all $i<t$, if $P_tit$ add to $A_t$ the contents of $A_i$.
Consider every $i,k\le t$. Let $n$ be the smallest number not added to any set
thus far. If $P_tik$, let $n\in A_k$. Else, let $n\notin A_k$.

Let us verify that this construction behaves as claimed. First, the sequence
$\{A_i\ |\ i\in\NN\}$ is clearly computable. Next, suppose $Pik$. If $i<k$, at
stage $k$
all elements placed thus far into $A_i$ will be placed into $A_k$. Thereon, at
any stage $t$ an element could only be placed into $A_i$ if there exists a $j$
with $P_tji$. However, as $P_t$ is transitive at that stage, $P_tjk$ and that
element is placed in $A_k$ as well. Next, suppose $k<i$. Clearly for $t>i$, any
element placed into $i$ at stage $t$ is also placed into $k$ as $P_t$ is
transitive. The only other time an element enters $i$ is at stage $i$ where we
add the contents of $A_j$ to $A_i$, for all $j$ where $P_i{ji}$. For
contradiction, let us assume this process violates our requirement that
$A_i\subseteq A_k$. Let $j$ be the smallest number such that $P_i{ji}$ but
$A_j\nsubseteq A_k$. Clearly $j>k$, otherwise we have already seen that all the
elements of $A_j$ would be in $A_k$. It is easy to see that this implies that at
stage $j$ there must have been a $j'$ with $P_jj'j$ (by transitivity, $P_jj'k$)
with $j'>k$ and $A_{j'}\nsubseteq A_k$. However as there are only finitely many
numbers between $k$ and $i$, eventually this would no longer be possible. As
such $Pik$ implies $A_i\subseteq A_k$.

Finally, suppose $\neg Pik$. This means there exists a stage $t$ with $\neg
P_tik$. At that stage we will place some element into $A_i$ (because clearly
$P_tii$), but not into $A_k$, meaning that $\neg Pik$ implies $A_i\nsubseteq
A_k$

We observe that the following is a decision procedure for determining whether
$n\in A_i$: perform the first $\max\{\lceil\sqrt{2n}\rceil,i\}$ steps of the
construction, and observe whether $n$ has been added to $A_i$ thus far. To see
that this is sufficient note that there are only two ways an element $n$ could
be added to $A_i$:
\begin{enumerate}
\item
$n$ has already been placed into some $A_k$, and $P_iki$ held at stage $i$. We
would
then add the contents of $A_k$ to $A_i$ at the beginning of that stage.
\item
At some stage $t$, $n$ was the smallest integer not placed in any set so far,
and was placed into $A_i$ because $P_tik$ for some $i$.
\end{enumerate}
The first case happens at stage $i$. In the second case, note that at every
stage $s$ we have $s$ sets under consideration, and each set induces the
introduction of a new integer. So the largest integer placed in any set at stage
$t$ is $\sum_{s\le t} s=t^2/2$, so $t\le\lceil\sqrt{2n}\rceil$.

While we can thus establish that we need to run no more than
$\max\{\lceil\sqrt{2n}\rceil,i\}$ steps of the construction, this tells us
nothing about how much time each step will actually take: at stage $t$ we must
first construct $P_t$, and that could take an arbitrarily long time. This is the
purpose of the padding function $h$.

Recall that the sequence $\{P_t\ |\ t\in\NN\}$ is computable. That is, there
exists a machine $M$ that on input $t$ returns the code of the machine
deciding $P_t$. As in Theorem~\ref{thm:quadratic_complete}, $h(n)$ is simply the
function that counts the number of steps $M$ takes on input $n$. As before, we
can assume that $h(n-1)\le h(n)$.

We can then define $X_i$ via its decision procedure as follows: to determine
whether $1^n0^k\in A_i$, first verify whether $k=h(n)$. If not, reject the
string. Otherwise perform the first $\max\{\lceil\sqrt{2n}\rceil,i\}$ steps of
the $\{A_i\ |\ i\in\NN\}$ construction and determine whether $n\in A_i$.

It is clear that this construction guarantees $n\in A_i$ if and only if
$1^n0^{h(n)}\in X_i$. It remains to verify that $X_i$ is polynomial time.

We can ignore the case where $\max\{\lceil\sqrt{2n}\rceil,i\}=i$ as it takes
constant time. As such, we will consider the time needed to perform the first
$t=\lceil\sqrt{2n}\rceil$ steps of the construction of $\{A_i\ |\ i\in\NN\}$.
More precisely, we will consider the amount of time needed to perform the
$t$th step, as performing all preceding step requires at
most $t$ times more work, which is sublinear in the length
of the input.

As in Theorem~\ref{thm:quadratic_complete}, the first step is to ensure the
input string is of the right form. That is, on input $1^n0^m$ we must first
verify whether $h(n)=m$. The argument is as before: begin calculating $h(n)$ by
simulating the machine which counts the number of steps $M$ takes on input $t$,
where $M$ is the machine that on input $t$ returns the code of the machine
deciding $P_t$. It should be clear whether the machine takes too little or too
many steps in time at most quadratic in $m$. If the string is of the wrong form,
reject it.

Given that the string is of the form $1^n0^{h(n)}$, it takes at most
$h(t)$ time to construct the required
approximation of $P$, whereas the input is of length $n+h(n)$. Following this we
copy the elements of $A_k$ into $A_{t}$ for all $k$ with
$P_{t}kt$: there are at most
$t$ such sets, and each set contains at most $n-1$ elements,
so the number of operations is polynomial in $n$.

Finally, we consider every $k\le t$, of which there are
$t$, and for every $j$ add a new element into $A_j$ if $P_tkj$: at most $t^2$
operations, or linear in $n$.
\end{proof}

We would like to complete the parallel with the equivalence relation case by
showing that the
embeddability of polynomial time trees is complete for $\PI 1$ preorders.
Unfortunately, embeddability in the natural sense of an injective function
preserving predecessor does not appear to be $\PI 1$ but rather $\DE 2$: we
could guess the node the root is mapped to and require that for every level
below that the embedding is preserved, or we could require that for every depth
of the preimage there exists a depth of the image enabling an embedding.

It is also worth noting that with this notion of an embedding bi-embeddability
would
not, in fact, imply isomorphism as illustrated in Figure~\ref{fig:bi-emb_trees}.
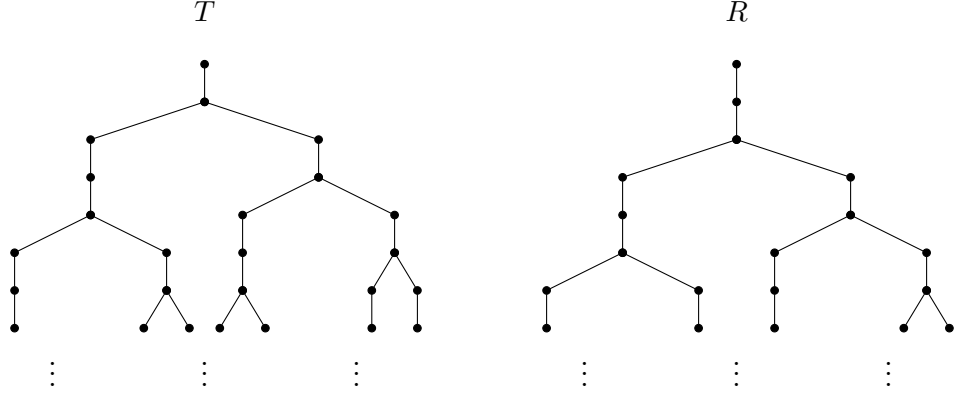
\begin{figure}
\begin{center}
\begin{tikzpicture}[dot/.style={draw,circle,minimum size=1mm,inner sep=0pt,outer
sep=0pt,fill=black}]

\node at (0,0.7) {$T$};

\path \foreach \x/\y/\k in { 0/0/0,  0/-0.5/1, -1.5/-1/2, -1.5/-1.5/3,
-1.5/-2/4, -2.5/-2.5/5, -2.5/-3/6, -2.5/-3.5/7}
      {coordinate [dot] (p\k) at (\x,\y)};
\draw \foreach \k/\i in {0/1,1/2,2/3,3/4,4/5,5/6,6/7}
      {(p\k)--(p\i)};

\path \foreach \x/\y/\k in { 0/-0.5/0,  1.5/-1/1, 1.5/-1.5/2, 0.5/-2/3,
0.5/-2.5/4, 0.5/-3/5, 0.2/-3.5/6}
      {coordinate [dot] (p\k) at (\x,\y)};
\draw \foreach \k/\i in {0/1,1/2,2/3,3/4,4/5,5/6}
      {(p\k)--(p\i)};

\path \foreach \x/\y/\k in { -1.5/-2/0, -0.5/-2.5/1, -0.5/-3/2, -0.8/-3.5/3}
      {coordinate [dot] (p\k) at (\x,\y)};
\draw \foreach \k/\i in {0/1,1/2,2/3}
      {(p\k)--(p\i)};

\path \foreach \x/\y/\k in { -0.5/-3/0, -0.2/-3.5/1}
      {coordinate [dot] (p\k) at (\x,\y)};
\draw \foreach \k/\i in {0/1}
      {(p\k)--(p\i)};

\path \foreach \x/\y/\k in { 0.5/-3/0, 0.8/-3.5/1}
      {coordinate [dot] (p\k) at (\x,\y)};
\draw \foreach \k/\i in {0/1}
      {(p\k)--(p\i)};

\path \foreach \x/\y/\k in { 1.5/-1.5/0, 2.5/-2/1, 2.5/-2.5/2, 2.2/-3/3,
2.2/-3.5/4}
      {coordinate [dot] (p\k) at (\x,\y)};
\draw \foreach \k/\i in {0/1,1/2,2/3,3/4}
      {(p\k)--(p\i)};

\path \foreach \x/\y/\k in { 2.5/-2.5/0, 2.8/-3/1, 2.8/-3.5/2}
      {coordinate [dot] (p\k) at (\x,\y)};
\draw \foreach \k/\i in {0/1,1/2}
      {(p\k)--(p\i)};

\node at (0,-4) {$\vdots$};

\node at (-2,-4) {$\vdots$};

\node at (2,-4) {$\vdots$};

\node at (7,0.7) {$R$};

\path \foreach \x/\y/\k in { 7/0/0, 7/-0.5/1,  7/-1/2, 5.5/-1.5/3, 5.5/-2/4,
5.5/-2.5/5, 4.5/-3/6, 4.5/-3.5/7}
      {coordinate [dot] (p\k) at (\x,\y)};
\draw \foreach \k/\i in {0/1,1/2,2/3,3/4,4/5,5/6,6/7}
      {(p\k)--(p\i)};

\path \foreach \x/\y/\k in { 7/-1/0,  8.5/-1.5/1, 8.5/-2/2, 7.5/-2.5/3,
7.5/-3/4, 7.5/-3.5/5}
      {coordinate [dot] (p\k) at (\x,\y)};
\draw \foreach \k/\i in {0/1,1/2,2/3,3/4,4/5}
      {(p\k)--(p\i)};

\path \foreach \x/\y/\k in { 5.5/-2.5/0, 6.5/-3/1, 6.5/-3.5/2}
      {coordinate [dot] (p\k) at (\x,\y)};
\draw \foreach \k/\i in {0/1,1/2}
      {(p\k)--(p\i)};

\path \foreach \x/\y/\k in { 8.5/-2/0, 9.5/-2.5/1, 9.5/-3/2, 9.2/-3.5/3}
      {coordinate [dot] (p\k) at (\x,\y)};
\draw \foreach \k/\i in {0/1,1/2,2/3}
      {(p\k)--(p\i)};

\path \foreach \x/\y/\k in { 9.5/-3/0, 9.8/-3.5/1}
      {coordinate [dot] (p\k) at (\x,\y)};
\draw \foreach \k/\i in {0/1}
      {(p\k)--(p\i)};

\node at (7,-4) {$\vdots$};

\node at (9,-4) {$\vdots$};
\node at (5,-4) {$\vdots$};

\end{tikzpicture}
\end{center}
\caption{Bi-embeddable, but non-isomorphic trees}
\label{fig:bi-emb_trees}
\end{figure}

Instead we will consider \emph{root-preserving} embeddability: the existence of
an injective function that maps the root to the root and respect predecessor.
This notion has a natural interpretation. $T_a$ is root-preserving embeddable
into $T_b$ if and only if $T_a$ is isomorphic to a set of prefixes of $T_b$.
That is, the nodes in $T_a$ can be renamed in such a way so that $T_a\subseteq
T_b$. This will also give us an alternative proof of Theorem~\ref{thm:Pi01_trees}.

\begin{theorem}
Root-preserving embeddability of polynomial time subtrees of $\{0,\dots,c\}^*$
is
complete for $\PI 1$ preorders.
\end{theorem}
\begin{proof}
The proof follows the same lines as Theorem~\ref{thm:Pi01_trees}. The desired
relation can be expressed as:
\begin{equation}
\{ (e,i):\forall d\ T_{(e,d)}\text{ is root-preserving embeddable in
}T_{(i,d)}\}.
\end{equation}

To obtain the theorem we will show how given indices $e$, $i$ of polynomial time
sets we can construct polynomial time subtrees of $\{0,1\}^*$, $T_e$ and $T_i$,
satisfying:
$$X_e\subseteq X_i\iff T_e\text{ is root-preserving embeddable in }T_i.$$

Let $X_e$ and $X_i$ be polynomial time sets. Let $1^x0\in T_e$ for every $x\in
X_e$ and $1^x0\in T_i$ for every $x\in X_i$. Let $1^y\in T_e,T_i$ for all $y$.
Suppose $X_e\subseteq X_i$. Note that our construction guarantees that
$T_e\subseteq T_i$, and as such the identity function is the required embedding.
Next, suppose that $f$ is a root preserving embedding of $T_e$ into $T_i$.
Observe that $f$ must map $1^x0$ to $1^x0$ as it is the only string of length
$x+1$ with no suffixes in the language. As such a string is in $T_i$ if and only
if $x\in X_i$, this establishes that $X_e\subseteq X_i$.\end{proof}
\begin{proposition}
If $T_a$ and $T_b$ are polynomial time subtrees of $\{0,\dots,c\}^*$ then $T_a$
and $T_b$ are root-preserving bi-embeddable if and only if $T_a\cong T_b$.
\end{proposition}
\begin{proof}
We only need to prove the proposition for finite trees, as the rest follows from
K\"{o}nig's lemma. One direction is obvious. For the other direction suppose
$T_a$ and $T_b$ are finite trees with $f$ a root-preserving embedding of $T_a$
into $T_b$ and $g$ a root-preserving embedding of $T_b$ into $T_a$. We claim
that $f$ is actually an isomorphism. As $f$ is already an injective, predecessor
preserving mapping, it is sufficient to show that $f$ is onto.

As $f$ and $g$ are injective functions from $T_a$ to $T_b$ and $T_b$ to $T_a$
respectively, it follows that $|T_a|\le |T_b|$ and $|T_b|\le |T_a|$. As
$|T_a|=|T_b|$, and both are finite sets, it follows that any one to one function
between them is also onto, thus $f$ is an isomorphism.
\end{proof}
\begin{corollary}
Another way of obtaining Theorem~\ref{thm:Pi01_trees}.
\end{corollary}

\section{The end of completeness}\label{sec:Pi0n}

In this section we present a modified proof of a result of Russell Miller and
Keng Meng Ng from
\cite{Ianovski2012}, concerning the non-existence of $\PI n$-complete
equivalence relations for $n\geq 2$. We then show that it is immediately
extendable to the $\DE n$ case, thereby completing our investigation of
equivalence relations complete under computable component-wise reducibility.

\begin{theorem}\label{thm:noPi0nrelation}
For $n\geq 2$, there are no $\PI n$-complete equivalence relations.
\end{theorem}
\begin{proof}
We present the proof for the case $n=2$. For a general $n$ we simply relativise
to sets
$\PI 2$ with an appropriate oracle.

Let $E$ be an arbitrary $\PI 2$ equivalence relation. Recall that $\{e\ |\
W_e\text{ is infinite}\}$ is complete for $\PI 2$ sets
(see, for instance, \cite{Soare1987} Theorem IV:3.2). It follows that there
exists an r.e.\ sequence of sets
$\{V[x,y]\ |\ x,y\in\NN\}$ satisfying:
$$V[x,y]\text{ is infinite}\iff (x,y)\in E.$$
We will construct an $F\in\PI 2$ satisfying $F\nleq E$ by constructing an r.e.\
sequence $\{U_{xy}\ |\ x,y\in\NN\}$ where:
$$U_{xy}\text{ is infinite}\iff (x,y)\in F.$$
To ensure $F$ is an equivalence relation we will let $U_{xx}=\NN$ for all $x$
and we will only explicitly construct sets $U_{xy}$ for $x<y$, taking
$U_{yx}=U_{xy}$ as given. We will verify that $F$ is transitive after
the construction.

We will give the procedure for enumerating $\{U_{xy}\ |\ x,y\in\NN\}$ by giving
a subroutine that we will dovetail over to eventually cover all natural numbers.
We will use $x_e,y_e$ and $z_e$ to denote the $3e+1$, $3e+2$ and $3e+3$
respectively. I.e., the $e$th number with remainder 1,2 and 0 on division by 3.
The subroutine is run on $x_e,y_e$ and $z_i$:
\begin{enumerate}
\item
Wait for $\varphi_e(x_e),\varphi_e(y_e)$ and $\varphi_e(z_i)$ to converge.
\item
If any of $\varphi_e(x_e),\varphi_e(y_e)$ or $\varphi_e(z_i)$ output the same
value, end the subroutine.
\item
Enumerate $V[\varphi_e(x_e),\varphi_e(z_i)]$ and
$V[\varphi_e(x_e),\varphi_e(y_e)]$, adding a new element to $U_{x_ez_i}$ for
every step of the computation. If an element enters
$V[\varphi_e(x_e),\varphi_e(z_i)]$ go to step 4. If an element enters
$V[\varphi_e(x_e),\varphi_e(y_e)]$ go to step 5.
\item
Enumerate $V[\varphi_e(y_e),\varphi_e(z_i)]$ and
$V[\varphi_e(x_e),\varphi_e(y_e)]$, adding a new element to $U_{y_ez_i}$ for
every
step of the computation. If an element enters
$V[\varphi_e(y_e),\varphi_e(z_i)]$ go to step 3. If an element enters
$V[\varphi_e(x_e),\varphi_e(y_e)]$ go to step 5.
\item
Restart the subroutine on $x_e,y_e,z_j$ where $z_j$ is the least multiple of 3
not yet used in any subroutine.
\end{enumerate}

As mentioned, to enumerate $\{U_{xy}\ |\ x,y\in\NN\}$ we dovetail the above
subroutine on the smallest $x_e,y_e$ and $z_i$ not yet used.

Let us now verify that $F$ is a $\PI 2$ equivalence relation as intended. That
$F$ is indeed $\PI 2$ is immediate from its definition via infinite r.e.\ sets.
 We
have already dealt with symmetry and reflexivity. We will demonstrate that $F$
has no equivalence class with more than 2 members, thus obtaining transitivity
automatically.

Observe that $(x_e,y_d)\notin F$ for any $d,e$ as we never place any element
into $U_{x_ey_d}$. We also claim that it is not possible for $(x,z_i),(x,z_j)\in
F$ for $i<j$ as we would only run the subroutine on the same $x$ and $z_j$ if we
reach step 5, and once we reach step 5 we would no longer place any new elements
into $U_{xz_i}$, so it could not be infinite. As such the only candidate for a
an equivalence class with more than two members is $\{x_e,y_e,z_i\}$ via
$U_{x_ez_i}$ and $U_{y_ez_i}$ being infinite. Observe that this is impossible:
if the subroutine ever reaches step 5 then $U_{x_ez_i}$ and $U_{y_ez_i}$ would
be finite. If it never reaches step 5, then either after a finite amount of
steps it stays at step 3 or 4 for ever, in which case either $U_{y_ez_i}$ or
$U_{x_ez_i}$ would be finite, or it must jump between 3 and 4 infinitely many
times. As such a jump happens only when an element enters
$V[\varphi_e(y_e),\varphi_e(z_i)]$ and $V[\varphi_e(x_e),\varphi_e(z_i)]$
receive a new element respectively, these sets must be infinite. That is not
possible because $E$ is an equivalence relation, so
$V[\varphi_e(x_e),\varphi_e(y_e)]$ would need to be infinite as well. But if an
element ever enters $V[\varphi_e(x_e),\varphi_e(y_e)]$, we would have gone to
step 5.

We are now ready to show that $F\nleq E$. Suppose for contradiction that it is:
there then exists an $f=\varphi_e$ such that $(x_e,z_i)\in F$ iff
$(\varphi_e(x_e),\varphi_e(z_i))\in E$. As we have argued above $(x_e,z_i)\in F$
only if after a certain stage the subroutine stays at step 3 for ever, but the
subroutine will leave step 3 if an element ever enters
$V[\varphi_e(x_e),\varphi_e(z_i)]$, so that set must be finite and
$(\varphi_e(x_e),\varphi_e(z_i))\notin E$. The same argument shows that
$(y_e,z_i)\in F$ would necessitate $(\varphi_e(y_e),\varphi_e(z_i))\notin E$.
The only remaining possibility is where $(x_e,z_i),(y_e,z_i)\notin F$ for any
$i$. This means an infinite number of subroutines on $x_e,y_e$ must have reached
step 5, but since such a subroutine would enter step 5 only if an element enters
$V[\varphi_e(x_e),\varphi_e(z_i)]$ or $V[\varphi_e(y_e),\varphi_e(z_i)]$,
either $(\varphi_e(x_e),\varphi_e(z_i))$ or $(\varphi_e(y_e),\varphi_e(z_i))$
must be in $E$, which contradicts the assumption that $\varphi_e$ is a
reduction.
\end{proof}

\begin{corollary}
There is no $\PI n$-complete preorder for $n\geq 2$.
\end{corollary}
\begin{proof}
The existence of such a preorder would, by Proposition~\ref{prop:preorders},
contradict the theorem.
\end{proof}

\begin{proposition}
For every $\PI 2$ equivalence relation $E$ there exists a $\DE 2$ equivalence
relation $F$ such that $F\nleq E$.
\end{proposition}
\begin{proof}
We will modify the construction in Theorem~\ref{thm:noPi0nrelation} to create a
computable sequence $\{U[x,y]\ |\ x,y\in\NN\}$ where:
$$U[x,y]\text{ is cofinite}\iff (x,y)\in F.$$
$$U[x,y]\text{ is finite}\iff (x,y)\notin F.$$
As $Fxy$ is then determined both by $U[x,y]$ being infinite and
$\overline{U[x,y]}$ being finite, it is both a $\PI 2$ and $\SI 2$ relation.

As before, we let $U[x,x]=\NN$ and $U[x,y]=U[y,x]$.Likewise $x_e,y_e$ and
$z_e$ denotes $3e+1$, $3e+2$ and $3e+3$
respectively.
The subroutine is run on $x_e,y_e$ and $z_i$:
\begin{enumerate}
\item
Wait for $\varphi_e(x_e),\varphi_e(y_e)$ and $\varphi_e(z_i)$ to converge. Place
a new element into $\overline{U[x_e,y_e]}$, $\overline{U[x_e,z_i]}$ and
$\overline{U[y_e,z_i]}$ for each step of the computation.
\item
If any of $\varphi_e(x_e),\varphi_e(y_e)$ or $\varphi_e(z_i)$ output the same
value, end the subroutine.
\item
Enumerate $V[\varphi_e(x_e),\varphi_e(z_i)]$ and
$V[\varphi_e(x_e),\varphi_e(y_e)]$, adding a new element to $U[x_e,z_i]$,
$\overline{U[x_e,y_e]}$ and $\overline{U[y_e,z_i]}$ for
every step of the computation. If an element enters
$V[\varphi_e(x_e),\varphi_e(z_i)]$ go to step 4. If an element enters
$V[\varphi_e(x_e),\varphi_e(y_e)]$ go to step 5.
\item
Enumerate $V[\varphi_e(y_e),\varphi_e(z_i)]$ and
$V[\varphi_e(x_e),\varphi_e(y_e)]$, adding a new element to $U[y_e,z_i]$,
$\overline{U[x_e,y_e]}$ and $\overline{U[x_e,z_i]}$ for
every
step of the computation. If an element enters
$V[\varphi_e(y_e),\varphi_e(z_i)]$ go to step 3. If an element enters
$V[\varphi_e(x_e),\varphi_e(y_e)]$ go to step 5.
\item
Place all remaining natural numbers into $\overline{U[x_e,z_i]}$ and
$\overline{U[y_e,z_i]}$. Restart the subroutine on $x_e,y_e,z_j$ where $z_j$ is
the least multiple of 3
not yet used in any subroutine.
\end{enumerate}

Note that when dovetailed this routine enumerates $\{U[x,y]\ |\ x,y\in\NN\}$ and
its complement, meaning it is a computable sequence as required.

It remains to see that if any set in the sequence is infinite it is in fact
cofinite. As we have seen in Theorem~\ref{thm:noPi0nrelation} the subroutine
cannot go between steps 3 and 4 for ever, if either of $U[x_e,z_i]$ or
$U[y_e,z_i]$ is infinite, the subroutine must eventually remain at that step. As
such, only a finite number of elements could have been placed in the complement.
\end{proof}
\begin{corollary}
For $n\geq 2$, there is no $\DE n$-complete equivalence relation.
\end{corollary}
\begin{proof}
We relativise the proof to an appropriate oracle.
\end{proof}

\chapter{Conclusion} \label{sec:Conclusion}

At the end of our enquiry we emerge with a fair idea of what the world of
complete equivalence relations and preorders looks like.

The $\SI n$ relations are relatively well behaved. There exist complete
equivalence relations for every level, and in the lower rungs these correspond
to natural mathematical notions. We further note that for the case of $\SI 4$,
Turing equivalence of r.e.\ sets is shown to be complete in \cite{Ianovski2012}.
The search
for examples in $\SI 5$ or higher is, perhaps, futile, as it would stretch the
spirit of
the word to refer to relations of quantifier depth five or more as ``natural".

In $\PI n$ we find a world starkly different from that in the standard theory of
$m$-reducibility. For $\PI 1$ a complete equivalence relation exists, but beyond
that any familiarity of degree structure breaks down. Even to find a $\PI n$
equivalence relation that dominates all $\DE n$ equivalence relations is
impossible. Thus the $\DE n$ case, too, lacks complete members beyond that shown
complete for $\DE 1$ in \cite{Gao2001}.

Our departure from the literature with an excursion into the study of preorders,
while offering greater generality, did not yield any additional surprises.
Wherever we investigated, preorders behaved in the same manner as equivalence
relations did. This suggests an open question to researchers in related fields:

\begin{question}
What can be said about (complexity, analytic, set-theoretic) classes that admit
complete equivalence relations, but not complete preorders, with respect to some
notion of component-wise reducibility?
\end{question}

Of course the relevance of this question is in itself contingent on whether a
parallel of Proposition~\ref{prop:preorders} holds. In a setting where classes
are not closed under intersection, the connection between equivalence relations
and preorders may be lost altogether.

The degree structure of equivalence relations under component-wise reducibility
is rich and complex, and largely beyond the scope of this work. Our contribution
concerns merely the maximum degrees - the complete equivalence relations, but we
hope that such as it is it may nevertheless shed some light for those who will
gaze deeper into the depths than we have managed in the course of this
thesis.
\bibliographystyle{alpha}
\bibliography{bibliography.bib}
\end{document}